\documentclass[journal]{IEEEtran}

\usepackage{tikz}
\usepackage{pgfplots}
\pgfplotsset{width=8cm,compat=1.9}
\usepgfplotslibrary{external}
\usepackage{caption}
\usepackage{subcaption}

\usepackage{amsmath}
\usepackage{amssymb}
\usepackage{setspace,csquotes}
\usepackage{enumerate}

\usepackage{amsthm}
\newtheorem{prop}{Proposition}
\newtheorem{theorem}{Theorem}
\newtheorem{lemma}{Lemma}
\newtheorem{deftn}{Definition}
\newtheorem{corol}{Corollary}
\newcommand{\ie}[0]{\textit{i.e.}}

\usepackage{verbatim}
\usepackage{diagbox}
\usepackage[fleqn,tbtags]{mathtools}
\usepackage[english]{babel}
\usepackage{setspace}
\usepackage{tabularx}
\usepackage{comment}
\usepackage{adjustbox}
\usepackage{colortbl}
\usepackage{color}
\usepackage{amsmath}
\usepackage{relsize}
\usepackage{cite}
\usepackage{multirow}

\definecolor{lightblue}{rgb}{0.57,0.85,0.96}
\definecolor{lightgray}{rgb}{0.75,0.75,0.75}
\definecolor{lightred}{rgb}{1,0.6,0.6}
\definecolor{lightgreen}{rgb}{0.77,1,0.28}
\definecolor{lightorange}{rgb}{1,0.74,0.1}
\definecolor{lightpurple}{rgb}{0.92,0.67,0.92}
\definecolor{lightbrown}{rgb}{1,0.74,0.4}
\newcolumntype{L}{>{\arraybackslash}p{0.8cm}}
\usepackage{stackengine,graphicx}

\addto\captionsenglish {}
\makeatletter \renewenvironment{proof}[1][\proofname] {\par\pushQED{\qed}\normalfont\topsep6\p@\@plus6\p@\relax\trivlist\item[\hskip\labelsep\bfseries#1\@addpunct{.}]\ignorespaces}{\popQED\endtrivlist\@endpefalse} \makeatother

\DeclareMathOperator\supp{supp}

\DeclarePairedDelimiter{\floor}{\lfloor}{\rfloor}

\DeclareMathOperator*{\argmin}{arg\,min}

\newcommand{\N}{\mathbb{N}}

\def\prehp(#1,#2){\ensuremath{  #1 \cdot #2 }}
\def\delequal{\mathrel{\ensurestackMath{\stackon[1pt]{=}{\scriptstyle\Delta}}}}

\newcolumntype{P}[1]{>{\centering\arraybackslash}p{#1}}
\usepackage{algorithm}
\usepackage{algorithmic}

\usepackage{graphicx}
\usepackage{psfrag}
\usepackage{dblfloatfix}

\begin{document}

\title{Pull or Wait: How to Optimize Query Age of Information}

\author{M. Emrullah Ildiz, Orhan T. Yavascan, Elif Uysal, and O. Tugberk Kartal \\
Dept. of Electrical and Electronics Engineering, METU, 06800, Ankara, Turkey \\
\{ emrullah.ildiz\_01, orhan.yavascan, uelif, tugberk.kartal\} @metu.edu.tr }
	\thanks{This work was supported in part by TUBITAK under Grant 117E215 and under Grant 119C028 and in part by Huawei. M. Emrullah Ildiz and Orhan T. Yavascan were supported by Turk Telekom and Turkcell, respectively, within the framework of 5G and Beyond Joint Graduate Support Programme coordinated by Information and Communication Technologies Authority. We thank Semanur Avsar for his assistance with simulations.}


\setlength{\belowdisplayskip}{2pt} \setlength{\belowdisplayshortskip}{2pt}
\setlength{\abovedisplayskip}{2pt} \setlength{\abovedisplayshortskip}{2pt}

    \vspace*{15pt}
    {\let\newpage\relax\maketitle}
\thispagestyle{empty} 
\vspace{-2cm}
\begin{abstract}
We study a pull-based status update communication model where a source node submits update packets to a channel with random transmission delay, at times requested by a remote destination node. The objective is to minimize the average query-age-of-information (QAoI), defined as the average age-of-information (AoI) measured at query instants that occur at the destination side according to a stochastic arrival process. In reference to a push-based problem formulation defined in the literature where the source decides to \textit{update or wait} at will, with the objective of minimizing the time average AoI at the destination, we name this problem the \textit{Pull-or-Wait} (PoW) problem. We provide a comparison of the two formulations: (i) Under Poisson query arrivals, an optimal policy that minimizes the time average AoI also minimizes the average QAoI, and these minimum values are equal; and (ii) the optimal average QAoI under periodic query arrivals is always less than or equal to the optimal time average AoI. We identify the PoW problem in the case of a single query as a stochastic shortest path (SSP) problem with uncountable state and action spaces, which has been not solved in previous literature. We derive an optimal solution for this SSP problem and use it as a building block for the solution of the PoW problem under periodic query arrivals. 
\end{abstract}

\begin{IEEEkeywords}
Age of information, Internet of things, pull-based communication, query age of information, status updates, stochastic shortest path problem, update or wait
\end{IEEEkeywords}

\IEEEpeerreviewmaketitle
\section{Introduction}\label{sect:introduction} 
\par The Internet of Things (IoT) paradigm has been gaining wide use in various sectors such as environmental monitoring\cite{IoT_environmental}, health and wellness\cite{IoT_health}, vehicular networks\cite{IoT_vehicular}, smart cities\cite{IoT_smartcities}, and so on. In many applications of these settings, a destination node seeks to have accurate information about a remote process measured by a sensor to utilize toward a computation. The received information packets by the destination node are not equally valuable: The value of the update packets highly depends on their timeliness.
\par As a metric to measure timeliness of update packets, the \textit{age-of-information} (AoI), or simply \textit{age}, has been introduced and studied in many different environments \cite{Kaul_2012_Howoftenshouldoneupdate, Yates_2021_AoISurvey, Sun_2019_AoIBook}. It is defined as the elapsed time since the generation of the latest received update packet. This definition makes it possible to measure the freshness of information for every time point at the destination node in contrast to the traditional metric, packet delay, that corresponds to the freshness of individual packets \cite{Kadota_2018_packetdelay}.
\par To minimize AoI in a status update system, a sensor or a source node can generate an update packet any time by its own will and immediately send it to a destination node through a communication channel; this is referred to as the \textit{generate-at-will} model \cite{Bacinoglu_2015, updateorwait, Yates_2012_Lazyistimely, Intro_1, Intro_2, Intro_3, Intro_5, Intro_6, Intro_7, Intro_8, Intro_9, Intro_10, Intro_11}. This model was introduced in \cite{Bacinoglu_2015} and further studied in \cite{updateorwait}. The problem formulation in \cite{updateorwait} is concerned with the source generating updates judiciously, to minimize the overall time average AoI over a channel that imposes a random transmission delay. In this paper, for brevity, we will refer to this formulation as the \textit{Update-or-Wait (UoW)} problem. In the UoW problem, the source controls the age by determining the submission times of the update packets to the channel. The approach of minimizing the time average age of information as an objective models a destination node that continuously utilizes the update packets; however in many IoT scenarios the application running at the destination side will utilize the information updates at certain times, rather than continuously \cite{DongAuDIoT, DongAuDRandorDeterm}. A policy that strives to keep the overall time average age at a minimum will not necessarily maintain minimal age at those utilization times. 

\par In this paper, we define an extension of the UoW problem, which is referred to as the \textit{Pull-or-Wait (PoW) problem}. In the PoW problem, the destination node requests an update packet from the source node in an effort to keep a low AoI at the next query instants, that are based on a stochastic arrival process. The \textit{query-age-of-information} (QAoI) is defined as the age values measured at query instants. The goal of the destination in the PoW problem is to determine optimal request points to minimize QAoI, knowing only the statistics of the channel delay and the query arrival processes. The following simple example reveals the difference between the UoW and PoW problems. 
\par \textit{Example 1: Consider an IoT monitoring system that requires an update packet every $4$ mseconds. Hence, the query instants are at times $4, 8, 12, \dots$ The transmission delay of this channel is constant at $1.5$ msec, but the requests for an update packet are assumed to arrive at the source node without any delay. The zero-wait policy is shown in \cite{updateorwait} to be the optimal update policy for the UoW problem when the transmission delays are constant. The evolution of the age of information under the zero-wait policy is shown in Figure \ref{fig:PullvsPush}. This policy results in a time average age of information equal to $2.25$ and performs one packet transmission per $1.5$ msec. On the other hand, a reasonable policy, which is later shown to be an optimal policy, for the PoW problem is that the destination node requests update packets at times $2.5, 6.5, 10.5, \dots$ as shown in Figure \ref{fig:PullvsPush}. As a result, this policy results in the average age of information at query instants equal to $1.5$ and performs one packet transmission per query, $4$ seconds.}
\begin{figure}[ht]
\centering
\psfrag{Q_1}[cc][cc]{$Q_1$}
\psfrag{Q_2}[cc][cc]{$Q_2$}
\psfrag{Q_3}[cc][cc]{$Q_3$}
\psfrag{5.5}[cc][cc]{$5.5$}
\psfrag{3}[cc][cc]{$3$}
\psfrag{1.5}[cc][cc]{$1.5$}
\psfrag{4}[cc][cc]{$4$}
\psfrag{Deltat}{$\Delta(t)$}
\psfrag{aaa}[cc][cc][1]{Optimal Policy of PoW Problem}
\psfrag{bbb}[cc][cc][1]{Optimal Policy of UoW Problem}
\includegraphics [scale = 1]{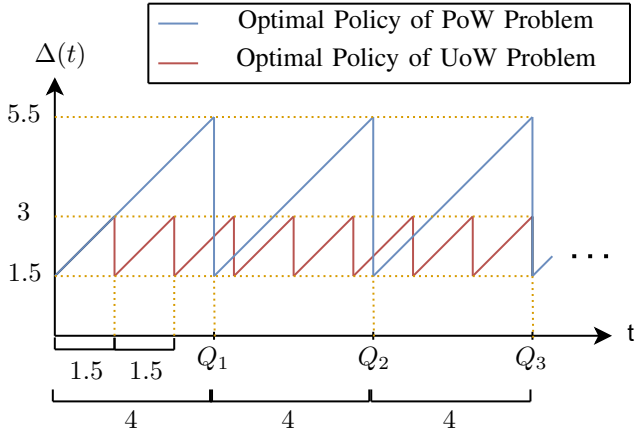}
\caption{Evolution of the Age of Information under the optimal policies of the UoW and PoW formulations in Example 1}
\label{fig:PullvsPush}
\end{figure}

This simple example points out a crucial distinction between the UoW and PoW problems. The PoW formulation uses the knowledge about utilization time \ie\, query instants to keep the AoI at the query instants much lower than that could be achieved in the UoW problem, while also reducing the number of transmissions. Hence, it is essential to comprehensively solve the PoW problem.

This paper aims to answer the following questions: How to optimally request update packets to minimize the age of information upon query instants at the destination? Under what conditions is the PoW model significantly advantageous over the UoW model? The following are the key contributions of this paper:
\begin{itemize}
    \item We define the PoW problem as a direct extension of the UoW problem formulated in \cite{updateorwait}. We show that under Poisson query arrivals, any optimal solution of the UoW problem is also an optimal solution of the PoW problem, achieving an equal age penalty (Proposition \ref{Poisson}).  We prove that for periodic queries the optimal average age penalty of the PoW problem is always less than or equal to that of the UoW problem with the same power constraint (Theorem \ref{superiorityTheorem}).

    \item We identify the PoW problem for a single query, referred to as \enquote{single query problem}, as a stochastic shortest path problem with uncountable state and action spaces, which has not been solved in previous literature, to the best of our knowledge. We show the existence of a deterministic policy that solves this problem (Proposition \ref{SuffandExis}) and characterize its first request point (Corollary \ref{CharacOfPolicy}). With the help of this characterization, we exhibit an explicit solution of the stochastic shortest path problem (Section \ref{sect:OptimalSolutionSection}).
    
    \item We employ the solution of the stochastic shortest path problem to construct a solution of the PoW problem under periodic query arrivals (Proposition \ref{InfRequest}). 
    
    \item We expand the results in \cite{chiariotti2021query} by relaxing three aspects of the system model: Our analysis allows a general channel delay distribution; a general age penalty function; and does not require a discount factor in the objective function.

\end{itemize}
\par The rest of the paper is organized as follows: In section \ref{sect:relatedwork}, we discuss some related work. In section \ref{sect:systemmodel}, we present the system model of the PoW problem. In section \ref{sect:problemformulationAnalysis}, we formulate the PoW problem and analyze it. In Section \ref{sect:SuperioritySection}, we prove that the solution of the PoW problem under periodic query arrivals always dominates that of the UoW problem. In section \ref{sect:NumericalResults}, we present numerical results to show the behavior of the solution in the PoW problem under different transmission delay processes. Finally, we conclude this paper in section \ref{sect:conclusion} by summarizing our contributions and discussing future directions.

\section{Related Work}\label{sect:relatedwork}
\par AoI has attracted a remarkable amount of interest \cite{Yates_2021_AoISurvey} and it has been applied to several different models and environments, such as enqueue-and-forward models \cite{Kam_2016_queue, Kam_2018_queue, Costa_2016_queue, Huang_2015_queue, Yates_2019_queue, Bedewy_2016_queue, Bedewy_2017_queue}, generate-at-will models \cite{updateorwait, Yates_2012_Lazyistimely, Intro_1, Intro_2, Intro_3, Intro_5, Intro_6, Intro_7, Intro_8, Intro_9, Intro_10, Intro_11}, random access environments \cite{Orhan_slotted, chen_2020_slotted, jiang_2018_slotted, yates_2017_slotted, shirin_2019_slotted}, and so on. Even though the age of information captures one semantic aspect of data, \ie\, the freshness of information, it is not sufficient for all applications. For example, the optimal policy that minimizes the MSE in the remote estimation of a Wiener process over a random delay channel is distinct from the age optimal policy as shown in \cite{sun2020sampling}. As a result, various suggestions for capturing the semantics of information have recently emerged \cite{uysal2021semantic, Bao2011TheorySemantic, Guler2018SemanticGame, Popovski2020}: the Age of Incorrect Information (AoII) extends the notion of fresh updates to that of fresh “informative” updates in \cite{Maatouk2020DefnAoII, chen2021minimizing, kriouile2021minimizing, Kam2020AoIIFeedback}. Other metrics such as the Urgency of Information (UoI) and the Age of Changed Information (AoCI) have been proposed in \cite{Zheng2020UoI} and \cite{Lin2020AoCI}, respectively.
\par The Query Age of Information (QAoI) is another metric that tries to capture the usefulness of an update packet with respect to an application more finely than the plain AoI. The QAoI is defined as the AoI measured at certain query instants, which represent the utilization times of the destination node in the application. This notion has been introduced in an independent set of works with different names such as Age upon Decision (AuD), Age of Effective Information (AoEI) \cite{Sang2017PullModel, Sang2021PullModel, DongAuDDefn, DongAuDGeneralArrivals, DongAuDIoT, DongAuDRandorDeterm, Yin2019AoEI}. The first works that suggest a pull-based communication model in the context of AoI are \cite{Sang2017PullModel, Sang2021PullModel}, where a user proactively requests update packets from multiple servers, but the authors minimize the plain AoI and do not take utilization time into account. A series of works \cite{DongAuDDefn, DongAuDGeneralArrivals, DongAuDIoT, DongAuDRandorDeterm} suggests AuD and studies a special case of the enqueue-and-forward model where a user utilizes upcoming update packets under a stochastic arrival process. This model leads the authors to measure the AoI at the utilization times. In \cite{Yin2019AoEI}, the authors study a multi-user information update system with Bernoulli update failures and suggest AoEI that measures the average AoI at the query instants. 
\par The works that are most relevant to this paper are \cite{updateorwait} and \cite{chiariotti2021query}. In \cite{updateorwait}, the authors consider a generate-at-will model to minimize time average AoI under a push-based communication model.  We extend \cite{updateorwait} to a pull-based communication model and modify the objective function with respect to the QAoI. In \cite{chiariotti2021query}, the authors suggest the QAoI and study a similar pull-based communication model. Unlike the packet erasure channel that is considered in \cite{chiariotti2021query}, we study more general channels that can have discrete, continuous, or mixed distributed transmission delays. In addition, we define an age penalty function $g(\Delta)$ to characterize the level of dissatisfaction for data staleness, where $g(.)$ can be any nonnegative, continuous, and nondecreasing function. This age penalty function enables us to simulate model-specific applications. Furthermore, we minimize the average age penalty at the query instants where there is no discount factor. In addition, we analytically compare the UoW and PoW problems under periodic and Poisson query arrival processes.

\section{System Model and Problem Definition}\label{sect:systemmodel}
\begin{figure}[t] 
\psfrag{Pull}[cc][cc][1]{\textbf{Pull (Request)}}
\psfrag{Source}[cc][cc][1]{\textbf{Source}}
\psfrag{Queue}[cc][cc][1]{\textbf{Queue}}
\psfrag{Channel}[cc][cc][1]{\textbf{Channel}}
\psfrag{Destination}[cc][cc][1]{\textbf{Destination}}
\centering
\includegraphics [scale = 1]{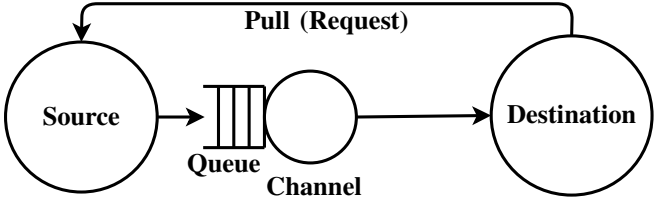}
\caption{System Model of the PoW Problem}
\label{fig:SystemModel}
\end{figure}

\par We consider a pull-based information update system depicted in Figure \ref{fig:SystemModel}, where a destination node is interested in information updates generated by a source node. The destination node requests an update packet from the source node according to an update policy. The request arrives at the source node without any delay. When a request occurs, the source node immediately generates an update packet and submits it to the channel. The channel induces a random delay between the source node and the destination node. The destination node should not request a new update packet when the previously requested update packet has not arrived at the destination node, because this will incur an unnecessary waiting time in the queue.
\par The update packets delivered to the destination node are utilized toward a computation. In this information update system, we assume that the destination node possesses a query arrival process that represents the utilization time of the upcoming update packets received from the source node. The destination node aims to minimize the average AoI at the query instants. As the destination node can recognize past states of the query arrival process, it requests update packets from the source node by taking account of not only the random delays induced by the channel but also the past states of the query process. 
\par Let the time that Update $j, j=1,2, \ldots$ is requested from the source, and submitted to the communication channel be denoted by $R_j$. Update $j$ is delivered to the destination node after a random transmission delay $Y_j$ at time $D_j= R_j+Y_j$. Then, the destination node requests Update $j+1$ at time $R_{j+1}$ after a waiting period $Z_j \in [0,M]$. This implies that $R_{j+1} = D_j + Z_j$. We assume that the transmission delay process, $\{Y_j\}_{j=0}^{\infty}$, is i.i.d. and takes values in a bounded range such that $\Pr(Y_j \in [B_{L}, B_{U}]) = 1$ where $B_L > 0$. On the other side, the query arrival process based on which the destination node utilizes the upcoming update packets is denoted as $\{Q_k, k=1,2,\ldots\}$. Please see Table \ref{NotationTable} for a summary of the notations used throughout the paper.

\par At any time $t$, let $U(t)$ denote the generation time of the update packet that has been most recently received by the destination node. Consequently,
\begin{equation}
    U(t) = \max\{R_j \colon D_j \leq t\}
\end{equation}
The age of information corresponding to this flow in the destination node at time $t$ is denoted by $\Delta(t)$, and is defined as:
\begin{equation}
    \Delta(t) = t-U(t)
\end{equation}

We also introduce an age penalty function, $g(\Delta)$, that represents the level of dissatisfaction for data staleness or the need for a new information update. This function is defined as $g \colon [0, \infty) \xrightarrow[]{} [0, \infty)$ and it is continuous, nonnegative, and nondecreasing. 
Our goal is to minimize the average age penalty at the time of queries by controlling the sequence of waiting periods,  $(Z_0, Z_1, \dots)$. Let $\pi \delequal (Z_0, Z_1, \dots)$ denote an update policy. A causal update policy determines the waiting period $Z_j$ based on the sequence $(Z_i)_{i=0}^{j-1}$, the random processes $\{Y_j\}_{j=0}^{\infty}$, $\{Q_k\}_{k=1}^{\infty}$, and their realizations before $D_j$. Let $\Pi$ be the set of all causal update policies. Then, the objective function is defined as the following:  
\begin{equation}\label{Problem2}
    \bar{h}_{opt} = \min_{\pi \in \Pi} \limsup_{n \xrightarrow[]{} \infty} \frac{E \big[\sum_{k=1}^n g(\Delta(Q_k))\big]}{n}
\end{equation}
Throughout the paper, we refer to this problem as the \textbf{Pull or Wait (PoW) problem}. We refer to the objective function of the PoW problem as the \textbf{query average age penalty}.
\subsubsection{UoW Problem}
\begin{figure}[t] 
\psfrag{ACK}[cc][cc][1]{\textbf{ACK}}
\psfrag{Source}[cc][cc][1]{\textbf{Source}}
\psfrag{Queue}[cc][cc][1]{\textbf{Queue}}
\psfrag{Channel}[cc][cc][1]{\textbf{Channel}}
\psfrag{Destination}[cc][cc][1]{\textbf{Destination}}
\centering
\includegraphics [scale = 1]{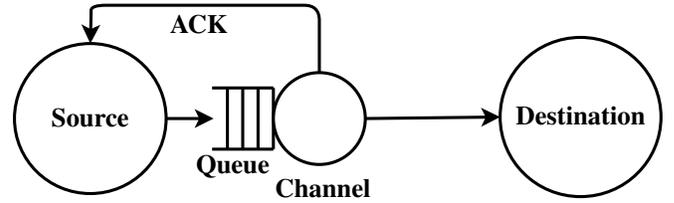}
\caption{System Model of the UoW Problem}
\label{fig:SystemModelUpdateorWait}
\end{figure}
\par In the system model that was studied in \cite{updateorwait} and that is depicted in Figure \ref{fig:SystemModelUpdateorWait}, the source node generates update packets and sends them directly to the destination node through the channel. Different from the system model of the PoW problem, the destination node does not request an update packet in an effort to minimize age penalty at the queries as there is no query in this system model. Instead, the source node submits update packets to the channel seeking to minimize the time average age penalty at the destination node. Therefore, the objective function is the following:
\begin{equation}\label{Problem1}
    \bar{g}_{opt} = \min_{\pi \in \Pi} \limsup\limits_{n\rightarrow \infty} \frac{E\big[\int^{D_n}_0 g(\Delta( t)) dt \big]}{E[D_n]} 
\end{equation}
Throughout the paper, we refer to this problem as the \textbf{Update or Wait (UoW) problem}. We refer to the objective function of the UoW problem as the \textbf{time average age penalty}.
\par Even though the example given in Section \ref{sect:introduction} leads to different update policies and average age penalty, the following proposition shows that there are special cases that both of the problems result in the same update policy and age penalty.
\begin{prop}\label{Poisson}
Let the query arrival process of a PoW problem be a Poisson process. For any transmission delay process, the optimal update policy that solves the UoW problem also solves the PoW problem with the same transmission delay process. Moreover, the optimal time average and query average age penalties are equal.
\end{prop}
\begin{proof}
The proof is provided in Appendix \ref{app:ProofOfPoisson}, and it is based on the \enquote{Poisson arrivals see time averages} property exhibited by the query process. 
\end{proof} 
\begin{table}[t]
\begin{center}
\begin{displaymath}
\begin{array}{|l l|}  \hline 
 R_j & \textrm{Request point of update } j  \\ 
  Y_j & \textrm{Transmission delay observed by update } j  \\  
 D_j & \textrm{Delivery time of update j}  \\  
 \multirow{2}{*}{$Z_j$} & \textrm{Waiting period after } D_j \textrm{ before requesting} \\
 &\textrm{the next update}\\ 
 A_j &  \textrm{Time elapsed from } D_j \textrm{ until the next query} \\
 B_L & \textrm{Lower bound on transmission delay}  \\  
 B_U & \textrm{Upper bound on transmission delay}  \\  
 Q_k & \textrm{The } k^{th} \textrm{ query time}  \\  
 T & \textrm{Time between queries in the deterministic case}  \\  
 \pi & \textrm{A causal update policy that consists of } Z_j's \\
 \Pi & \textrm{Set of all causal policies } \pi \\ 
 \multirow{2}{*}{$\Pi_{SD}$} & \textrm{Set of all causal and stationary deterministic} \\ & \textrm{policies } \pi \\
  \multirow{2}{*}{$\bar{h}_{opt}$}  & \textrm{Optimal query average age penalty in PoW}\\& \textrm{problem} \\ 
  \multirow{2}{*}{$\bar{g}_{opt}$}  & \textrm{Optimal time average age penalty in UoW}\\ &\textrm{problem} \\ \hline
\end{array}
\end{displaymath}
\caption{Summary of Notations}\label{NotationTable}
\vspace{-0.5cm}
\end{center}
\end{table}
\section{Problem Formulation and Analysis}\label{sect:problemformulationAnalysis}
In this section, we first analyze the PoW problem under a specific case of single query. Let $Q>0$ be the time at which the query occurs. For this case, Problem \eqref{Problem2} reduces to:
\begin{equation}\label{MinSingleQuery}
   \bar{h}^{one}_{opt}(Q) = \min_{\pi \in \Pi} E\big[ g(\Delta(Q))\big]
\end{equation}

\par Henceforth, we will refer to Problem \eqref{MinSingleQuery} as the \enquote{single query problem}. As we will show in the rest of this section, the solution of the single query problem will be a building block of the solution of the PoW problem, given in \eqref{Problem2}, under periodic query arrivals. 

\par The single query problem belongs to the class of stochastic shortest path problems with uncountable state and action spaces. The state of the problem at stage $j$ is the pair of the remaining time from the delivery point of Update $j$ until the query and the current age at the delivery point of Update $j$, $(Q-D_j, \Delta(D_j))$\footnote{It is shown in Proposition \ref{SuffandExis} that there exists an optimal policy of the single query problem in which $Z_j$ is determined as a function of $Q-D_j$ and $\Delta(D_j)$. As a result, the single query problem can be minimized in the set of deterministic policies. When $Z_j$ is determined as a function of $Q-D_j$ and $\Delta(D_j)$, the pair $(Q-D_j, \Delta(D_j)), j \geq 0$ forms a Markov chain because $\Delta(D_j) = Y_j$, $Y_j$'s are i.i.d., and $Q-D_{j+1} = Q-D_j-Y_j-Z_j$.}. The random disturbance and the control action at stage $j$ are $Y_j$ and $Z_j$, respectively. The absorbing state occurs at stage $j$ when $Q-D_j \leq 0$. State transitions that do not end in the absorbing state are costless. The cost of reaching the absorbing state from a state $(Q-D_j, \Delta(D_j))$ where $Q-D_j > 0$ is $g(Q-D_j+\Delta(D_j))$. This problem class is introduced in \cite{Bertsakas_1991} for a finite state space, compact action space, a transition kernel that is continuous for all actions, under the assumption that an optimal policy must be proper (\ie\, reachability of the termination state in a finite expected time). \cite{Bertsakas_2018} relaxes the assumptions of \cite{Bertsakas_1991} such that the state and action spaces are arbitrary, the transition kernel does not need to be continuous, but the space of the random disturbance is countable. A related problem class is introduced by \cite{Pliska_1978} as transient Markov decision problems with solutions that are \textit{transient policies} (similar, but not identical, to \textit{proper policies}), general state and action spaces, and continuous transition kernel. \cite{James_Collins_2006} further relaxes the assumptions of \cite{Pliska_1978} to the existence of non-transient policies, but keeps the assumption about the continuity of the transition kernel \cite[Assumption~1b]{James_Collins_2006}. None of these results are directly applicable to the single query problem because in our problem the random disturbance $Y_j$ may not come from a countable set and the transition kernel is not restricted to be continuous especially when the random disturbance $Y_j$ has a mixed distribution. 

\par In the rest of this section, we will show the existence of a deterministic optimaş policy for the single query problem, and characterize its first request point in Section \ref{sect:SingleQueryCase}. With the help of this characterization, we will reformulate the PoW problem under periodic query arrivals in terms of the single query problem in Section \ref{sect:PeriodicQueryCase}. Finally, we will provide a complete solution of the single query problem in Section \ref{sect:OptimalSolutionSection}, which concludes the solution of the PoW problem in \eqref{Problem2} under periodic query arrivals.
\vspace{-0.3cm}
\subsection{Existence of a Deterministic Optimal Policy for the Single Query Problem}\label{sect:SingleQueryCase}
\par In this subsection, we first show that there exists an optimal policy, $\pi_1^{opt}$, for the single query problem, that is a deterministic policy. Then, we define the \textit{border point} of $\pi_1^{opt}$ for a query arriving at time $Q$, denoted as $Q^{BP} \in [Q-3B_U, Q-B_U]$. We prove that $Q^{BP}$ is an optimal request point under the policy $\pi_1^{opt}$ for every delivery point $D_j$ satisfying $D_j < Q-3B_U$. This property will help us transform the solution of the single query problem into a solution of the PoW problem under periodic query arrivals. 
\par At any delivery point $D_j$, an optimal update policy seeks to find a request point $R_{j+1}$ to minimize the expected age penalty at the query. To express the expected age penalty at the query in terms of a request point $R_j$, we define the $G^{\pi}_R$ function. In addition to the $G^{\pi}_R$ function, we define the $G^{\pi}_D$ function to express the expected age penalty at the query in terms of a delivery point $D_j$ as the following:
\begin{deftn}
 For a given query $Q$, let $R_j$ and $D_j$ be any request and delivery points, respectively. $G^{\pi}_R \colon [0, \infty) \times [0, \infty) \rightarrow [0, \infty)$ and $G^{\pi}_D \colon [0, \infty) \times [0, \infty) \rightarrow [0, \infty)$ are defined as follows:
 \begin{equation}
 \begin{aligned}
     G^{\pi}_R\bigg(Q-R_j, \Delta(R_j)\bigg) \delequal E\bigg[g(\Delta(Q)) \bigg| \pi \text{ is applied,} \, \\
     R_j \text{ is a request point,} \text{AoI at } R_j \text{ is } \Delta(R_j)  \bigg] 
 \end{aligned}
 \end{equation}
 \begin{equation}
     \begin{aligned}
      G^{\pi}_D\bigg(Q-D_j, \Delta(D_j)\bigg) \delequal E\bigg[g(\Delta(Q)) \bigg| \pi \text{ is applied,}\, \\
      D_j \text{ is a delivery point}, \text{AoI at } D_j \text{ is } \Delta(D_j) \bigg]
     \end{aligned}
 \end{equation}
These expectations are taken over the possible transmission delays and the waiting period decisions by the policy $\pi \in \Pi$.
\end{deftn}
\par It will be shown in Proposition \ref{SuffandExis} that the information of the remaining time until the query $Q-D_j$ and the AoI at the delivery point $\Delta(D_j)$ are sufficient statistics to determine an optimal waiting period. This implies that the minimization of the single query problem can be performed by only considering the set of causal policies that determines the waiting period $Z_j$ based on $Q-D_j$ and $\Delta(D_j)$. Therefore, there is no need to explicitly provide the sequences of $(Y_i)_{i=0}^{j}$ and $(Z_i)_{i=0}^{j-1}$ for the functions $G^{\pi}_R$ and $G^{\pi}_D$.
\par The two functions have a chain relationship with each other. When the destination node requests an update packet from the source node at $R_j$, Update j is delivered to the destination node after a random transmission delay $Y_j$ at time $D_j = R_j + Y_j$. Hence, $\Delta(D_j) = Y_j$. If the delivery occurs before the query i.e.\, $Q-R_j-Y_j \geq 0$, the expected age penalty can be represented with the function $G^{\pi}_D$. If $Q-R_j-Y_j < 0$, the AoI at the query is $Q-R_j+\Delta(R_j)$ for sure. This relationship can be written as follows:
\begin{equation}\label{GR}
    \begin{split}
        G^{\pi}_R\bigg(Q - &R_j, \Delta(R_j)\bigg) \\
     & = E\bigg[ G^{\pi}_D(Q-R_j-Y_j, Y_j) \bigg| Y_j \leq Q-R_j \bigg] \\
   &  \hspace{3.45cm} \times \Pr(Y_j \leq Q-R_j) \\
     & + g\big(Q-R_j+\Delta(R_j)\big) \times \Pr(Y_j > Q-R_j)
    \end{split}
\end{equation}
This expectation is taken over possible transmission delays.
\par On the other hand, when the update packet is delivered to the destination node at $D_j$, the destination node waits for a duration $Z_j$ to request a new update packet. Hence, the request point is $Q-D_j-Z_j$, and the AoI at the request point is $\Delta(D_j) + Z_j$. When the request point is before the query i.e, $Q-D_j-Z_j \geq 0$, the expected age penalty at the query can be represented with the function $G^{\pi}_R$. When $Q-D_j-Z_j < 0$, the AoI at the query is $Q-D_j+\Delta(D_j)$ for sure. This relationship can also be written as follows:
\begin{equation}\label{GD}
    \begin{split}
        &G^{\pi}_D  \bigg(Q - D_j, \Delta(D_j)\bigg)  \\
    &= E\bigg[ G^{\pi}_R(Q-D_j-Z_j, \Delta(D_j) + Z_j) \bigg| Z_j \leq Q-D_j \bigg]  \\
    &\hspace{4.75cm} \times \Pr(Z_j \leq Q-D_j) \\
    &+g\big(Q-D_j+\Delta(D_j)\big) \times \Pr(Z_j > Q-D_j) 
    \end{split}
\end{equation}
This expectation is taken over possible waiting periods that are determined by the policy $\pi$ in order to take randomized policies into account. 
\par Now, we move on to obtain a deterministic optimal policy of the single query problem. The optimal age penalty in this problem can be achieved in a special subset of $\Pi$. In the next proposition, we prove this in detail. 
\begin{deftn}
\begin{itemize}
    \item []
    \item A policy $\pi \in \Pi$ is said to be a stationary and deterministic policy if there exists decision function $z:[0, \infty) \times [0, \infty) \rightarrow [0, M]$ such that $Z_j = z(Y_j, Q-D_j)$ for $j = 0,1, \dots$
    \item The set of all stationary and deterministic policies is denoted as $\Pi_{SD}$.
\end{itemize}
\end{deftn}
\begin{prop}\label{SuffandExis}
If the transmission delay process $\{Y_j\}_{j=0}^{\infty}$ is i.i.d. such that $\Pr(Y_j \in [B_L, B_U] ) = 1$, $M < \infty$, and the penalty function $g$ is continuous, non-negative, and non-decreasing, then there exists a deterministic update policy that is optimal for the single query problem.
\end{prop}
\begin{proof}
In the proof, we need to use the extended version of the functions $G^{\pi}_R$ and $G^{\pi}_D$ that must include the sequences of $(Y_i)_{i=0}^j$ and $(Z_i)_{i=0}^{j-1}$ in order to cover all possible causal update policies. Hence, they are $G^{\pi}_D(Q-D_j, \Delta(D_j), (Y_i)_{i=0}^j, (Z_i)_{i=0}^{j-1})$ and $G^{\pi}_R(Q-R_j, \Delta(R_j), (Y_i)_{i=0}^{j-1}, (Z_i)_{i=0}^{j-1})$. Let us map each $Q-D_j$ to a natural number $n$ satisfying $(n-1)B_L \leq Q-D_j < nB_L$. We perform discrete induction on $n$. The proposition is first proved for every $j$, $(Y_i)_{i=0}^{j}$, and $(Z_i)_{i=0}^{j-1}$ that satisfy $(n-1)B_L \leq Q-D_j < nB_L$ when $n=1$. Then, the proposition is assumed to be correct when $n=2,3, \dots , K$ where $K$ is an arbitrary natural number. Finally, it is proved when $n=K+1$. The details are available in Appendix \ref{app:SuffandExis}.
\end{proof}
According to the previous proposition, there exists a deterministic optimal update policy $\pi_1^{opt} \in \Pi_{SD}$ that decides waiting periods based on the values of $Q-D_j$ and $\Delta(D_j)$ for every $j$, $(Y_i)_{i=0}^j$, and $(Z_i)_{i=0}^{j-1}$. Interestingly, for some specific values of $Q-R_j$, the expected age penalty at the query may not depend on the value of $\Delta(R_j)$. For example, when the destination node is supposed to request an update packet from the source node before $Q-B_U$, the requested update packet must reach the destination node before the query. This is because the transmission delay can be at most $B_U$. Therefore, the AoI at the request point cannot affect the expected age penalty at the query. The next proposition proves this in detail.  
\begin{prop}\label{AoISubNotAffect}
If the elapsed time since a request point until the query is greater than $B_U$, then the AoI at the request point does not affect the expected age penalty at the query under a deterministic policy.
\end{prop}
\begin{proof}
This proposition is an immediate result of \eqref{GR}. If $Q-R_j \geq B_U$, then $Y_{j} \leq Q-R_j$ for sure. Therefore, \eqref{GR} becomes
\begin{equation}
     G^{\pi}_R\bigg(Q - R_j, \Delta(R_j)\bigg) =  E\bigg[ G^{\pi}_D\big(Q-R_j-Y_j, Y_j\big) \bigg]
\end{equation}
As the transmission delay process is i.i.d. and $\Delta(R_j) = Y_{j-1} + Z_{j-1}$, $\Delta(R_{j})$ does not affect $G^{\pi}_D(Q-R_j-Y_j, Y_j)$ when $Q-R_j$ is given. Hence, the proof is completed. Note that this property is valid for every $\pi \in \Pi_{SD}$.
\end{proof}
As a result of previous proposition, we can modify the function $G^{\pi_1^{opt}}_R$ when $\pi_1^{opt}$ is a deterministic optimal policy and $Q-R_j$ is greater than or equal to $B_U$. Hence, for every request point $R_j$ and its AoI $\Delta(R_j)$ satisfying $Q-R_j \geq B_U$, we redefine the $G^{\pi_1^{opt}}_R$ function with one argument as the following:
\begin{equation}
    G^{\pi_1^{opt}}_R(Q-R_j) \delequal G^{\pi_1^{opt}}_R\bigg(Q-R_j, \Delta(R_j)\bigg) 
\end{equation}
\par For a given query $Q$ and a deterministic optimal policy $\pi_1^{opt}$, let us define its border point $Q^{BP}$ that satisfies the following:
\begin{equation}\label{BorderPoint}
    G^{\pi_1^{opt}}_R(Q-Q^{BP}) = \inf_{R_j\colon R_j \leq Q-B_U} G^{\pi_1^{opt}}_R(Q-R_j)
\end{equation}
\par In the next proposition, we show the existence of a border point. Then, we specify one of these points as the border point.
\begin{prop}\label{BorderPointProp}
Let $D_j^*$ be a specific delivery point satisfying $D_j^* = Q-3B_U$ and $Y_j^* = B_L$. The request point $R_{j+1}^*$ determined by a deterministic optimal policy $\pi_1^{opt}$ is a border point for the query $Q$. We designate $R_{j+1}^*$ as the \enquote{selected} border point.
\end{prop}
\begin{proof}
In the proof, we first prove that the request must occur by the time $Q-B_U$ i.e, $R_{j+1}^* \leq Q-B_U$.  This ensures that $R_{j+1}^*$ is in the intended interval of \eqref{BorderPoint}. Then, we show that the optimal request point, $R_{j+1}^*$, attains the infimum in \eqref{BorderPoint}. The details are in Appendix \ref{app:BorderPointProp}.
\end{proof}
\par In the rest, for brevity, we will refer to the selected border point as the border point. The exact location, $Q^{BP}$, of the border point depends on the exact time of the query and the optimal policy $\pi_1^{opt}$. This is because the border point is specified as the request point that is determined by $\pi_1^{opt}$ when the delivery point is $Q-3B_U$ and the age at the delivery point is $B_L$. Hence, the border point can be considered as a function of a query $Q$ and a deterministic optimal policy $\pi_1^{opt}$. Nevertheless, there is a special property of the border point concerning the relation between $Q$ and $Q^{BP}$, proved in the following corollary:    
\begin{corol}\label{DiffRandRBP}
The time duration between a query and its border point does not depend on the exact time of the query for a given deterministic optimal policy.
\end{corol}
\begin{proof}
This corollary is an immediate result of Proposition \ref{SuffandExis} and the definition of $R_{j+1}^*$. The request point $R_{j+1}^*$ determined by a deterministic optimal policy $\pi_1^{opt}$ is the border point when $D_j^* = Q-3B_U$ and $\Delta(D_j^*) = Y_j^* = B_L$ regardless of the exact time of the query. The optimal waiting period at $D_j^*$ is solely determined by $\pi_1^{opt}$ based on $Q-D_j^*$ and $\Delta(D_j^*)$ by  Proposition \ref{SuffandExis}. As $Q$ changes, $Q-D_j^*$ and $\Delta(D_j^*)$ do not change. Hence, $Z_j^*$ does not change. As $R_{j+1}^* = Q-3B_U+Z_j^*$, the proof is completed.
\end{proof}

We next prove in Lemma \ref{BorderPointSub} that if a delivery point occurs before $Q-3B_U$, then it is optimal to wait until the border point to place a request.

\begin{lemma}\label{BorderPointSub}
Let $Q^{BP}$ be the border point of a query $Q$ and a deterministic optimal policy $\pi_1^{opt}$. Then, for any delivery point $D_j$ satisfying $D_j< Q-3B_U$, the border point $Q^{BP}$ is an optimal request point under the policy $\pi_1^{opt}$.
\end{lemma}
\begin{proof}
To reach contradiction, suppose that the claim is false. Then, there exists a delivery point $D_j \in [0, Q-3B_U)$ and an AoI at the delivery $\Delta(D_j)$ such that the request point $R_{j+1}$ determined by a deterministic optimal policy $\pi_1^{opt}$ satisfies the following: $G^{\pi^{opt}_1}_R(Q-Q^{BP}) > G^{\pi^{opt}_1}_R(Q-R_{j+1}, \Delta(R_{j+1}))$\footnote{As $R_{j+1}$ can be in the interval $[Q-B_U, Q]$, the $G^{\pi^{opt}_1}_R$ function should be written with AoI argument.}. By \eqref{BorderPoint}, $R_{j+1}$ cannot be in the interval $[0, Q-B_U)$. By Lemma \ref{RequestUntilQ-3BU} that is given in Appendix \ref{app:BorderPointProp}, $R_{j+1}$ cannot be in the interval $[Q-B_U, Q]$ as well. This completes the proof.
\end{proof}
\begin{corol}\label{CharacOfPolicy}
There exists a deterministic optimal policy $\pi_1^{opt}$ for a given query $Q$ satisfying $Q> 3B_U$ such that the first request point is the border point.
\end{corol}
\begin{proof}
This is an immediate result of Lemma \ref{BorderPointSub} and the designation of the border point in Proposition \ref{BorderPointProp}.
\end{proof}
\begin{corol}\label{honeopt}
If $Q > 3B_U$, $\bar{h}^{one}_{opt}(Q)$ is independent of the exact time of the query $Q$. In other words, we can define $\bar{h}^{one}_{opt}$ as the following:
\begin{equation}
    \bar{h}^{one}_{opt} \delequal \bar{h}^{one}_{opt}(Q) = G^{\pi_1^{opt}}_R(Q-Q^{BP}) 
\end{equation}
where $\pi_1^{opt}$ is a deterministic optimal policy and $Q^{BP}$ is their border.
\end{corol}
\begin{proof}
From Corollary \ref{CharacOfPolicy}, there exists a deterministic optimal policy $\pi_1^{opt}$ whose first request point is the border for a given query $Q$ satisfying $Q> 3B_U$. This means that $ \bar{h}^{one}_{opt}(Q) = G^{\pi_1^{opt}}_R(Q-Q^{BP})$. Furthermore, the time duration between $Q-Q^{BP}$ does not change when $Q$ is shifted by Corollary \ref{DiffRandRBP}. Hence, the expected age penalty at the border point for any $Q>3B_U$ is the same because the destination node can request an update packet at the border point under an optimal policy. As a result, we can define $\bar{h}^{one}_{opt} \delequal \bar{h}^{one}_{opt}(Q)$. This completes the proof. 
\end{proof}
\par Thus far, we have shown the existence of an optimal policy $\pi_1^{opt}$ that has two important properties:
\begin{itemize}
    \item $\pi_1^{opt}$ is a deterministic optimal policy that decides the waiting period at $D_j$ solely based on the $Q-D_j$ and $\Delta(D_j)$.
    \item The first request point of the policy $\pi_1^{opt}$ is in the interval $[Q-3B_U, Q-B_U]$.
\end{itemize}
These two properties enable us to transform the optimal update policy of the single query problem into an optimal update policy of the PoW problem under periodic query arrivals.
\subsection{Periodic Sequence of Queries}\label{sect:PeriodicQueryCase}
\par In this subsection and next subsection, we assume that the query arrival process $\{Q_k\}_{k=1}^{\infty}$ is deterministic and periodic with $T$. Let $Q_k = kT$, for $k = 1,2, \dots$. Furthermore, we assume that $T>4B_U$.\footnote{Considering the delay in many practical communication links is expected to be much lower than the query period for typical applications, this assumption is not restrictive for many practical cases of interest.} Based on these assumptions, we construct an optimal update policy $\pi^{opt}$ for a periodic sequence of queries in the next proposition. Then, we point out the properties of the update policy $\pi^{opt}$ based on the next proposition.  
\begin{prop}\label{InfRequest}
If the transmission delay process $\{Y_j\}_{j=0}^{\infty}$ is i.i.d. such that $\Pr(Y_j \in [B_L, B_U] ) = 1$ and the query arrival process, $\{Q_k\}_{k=1}^{\infty}$, is deterministic and periodic with $T> 4B_U$, then $\bar{h}_{opt}$ is equal to $\bar{h}^{one}_{opt}$.
\end{prop}
\begin{proof}
It is clear that $\bar{h}^{one}_{opt} \leq \bar{h}_{opt}$. Otherwise, it would contradict the optimal solution of the single query problem. Therefore, it is enough to construct an update policy $\pi^{opt}$ achieving $\bar{h}^{one}_{opt}$ of expected age penalty for the periodic sequence of queries.
\par Let $\pi_1^{opt}$ be the optimal policy of the single query problem characterized in Corollary \ref{CharacOfPolicy}. Let $Q_i^{BP}$ be the border point of $Q_i$ and $\pi_1^{opt}$. From the starting point, $\pi^{opt}$ can follow $\pi_1^{opt}$ between $[0, Q_1]$. This can be performed because $\pi_1^{opt}$ decides to wait until $Q_1^{BP}$ and $Q_1^{BP} \geq Q_1 - 3B_U > 0$. From Corollary \ref{honeopt}, the expected age penalty at $Q_1$ is $G^{\pi_1^{opt}}_R(Q_1 - Q_1^{BP})$. As the policy $\pi^{opt}$ follows $\pi_1^{opt}$ until the point $Q_1$, the channel must be idle before $Q_1+B_U$ as the transmission delay can be at most $B_U$. When the channel is idle, $\pi^{opt}$ can follow $\pi_1^{opt}$ again, but this time the policy is performed for the query $Q_2$. The act of following the policy $\pi_1^{opt}$ is possible because $Q_2^{BP} \geq Q_2-3B_U > Q_1 + B_U$. Hence, the expected age penalty at $Q_2$ is $G^{\pi_1^{opt}}_R(Q_2-Q_2^{BP})$ by Corollary \ref{honeopt}. For the remaining queries $Q_3, Q_4, \dots$, it can be replicated similar to $Q_2$. Then, the expected age penalty at every query $Q_k$ is  $G^{\pi_1^{opt}}_S(Q_k-Q_k^{BP})$. From Corollary \ref{DiffRandRBP}, all of the expected age penalties are equal to $\bar{h}^{one}_{opt}$.
\end{proof}
\par The previous proposition allows us to decouple the immediate next query from the set of all the queries while constructing an optimal policy $\pi^{opt}$ for the PoW problem under periodic query arrivals. As a result, the update policy $\pi^{opt}$ takes only the immediate next query into account. This decoupling property enables us to solve the PoW problem without a discount factor. The next corollary presents another result of the decoupling property.   
\begin{corol}\label{StatandDet}
Let $A_j \delequal Q - T \left \lfloor \frac{D_j}{T} \right \rfloor$ that represents the remaining time until the next query at a delivery point $D_j$. The update policy $\pi^{opt}$ constructed in Proposition \ref{InfRequest} is a stationary and deterministic policy, which is a function of $\Delta(D_j) = Y_j$ and $A_j$.
\end{corol}
\begin{proof}
The update policy $\pi^{opt}$ is a repetitive employment of the update policy $\pi_1^{opt}$, that is characterized in Corollary \ref{CharacOfPolicy}. Therefore, $\pi^{opt}$ possesses all the properties of $\pi_1^{opt}$. As $\pi_1^{opt}$ is solely determined based on $Q-D_j$ and $\Delta(D_j)$ by Proposition \ref{SuffandExis}, $\pi^{opt}$ is stationary and deterministic function of $A_j$ and $\Delta(D_j) = Y_j$. 
\end{proof}
\par Note that we prove in Corollary \ref{StatandDet} that the constructed policy $\pi^{opt}$ is a stationary and deterministic policy, which is a function of $A_j$ and $Y_j$. We also show in Proposition \ref{InfRequest} that the optimal update policy for the PoW problem under periodic query arrivals turns out to \textit{myopic} in the sense that at any delivery point, the decision about the optimal waiting time does not depend on future queries other than the immediate next one. Therefore, what remains to solve the PoW problem is to find an optimal policy for the single query problem, and apply it at each consecutive query interval.

\subsection{Explicit Solution of PoW Problem}\label{sect:OptimalSolutionSection}

\par In the previous subsection, we exploited the decoupling property Proposition \ref{InfRequest} to show that one can construct a solution of the PoW problem under periodic query arrivals through employing a sequence of deterministic policies that each solve the single query problem. In this subsection, we provide an explicit solution of the single query problem by generating a sequence of update policies that are solutions of stochastic shortest path problems with finite state and action spaces obtained by quantization. Then, we show that the sequence of update policies converges to an optimal policy of the single query problem with increasingly fine quantization. The quantization argument is given next.
\par We divide the real line interval $[0, Q]$ into $N$ equal sub-intervals, and define two new transmission delay processes: 
\begin{enumerate}
    \item \underline{Upper Quantized Transmission Delay Process:}  If a transmission delay $Y_j$ occurs with a probability in a transmission delay process, the transmission delay is quantized to $\frac{Q}{N}\left \lceil \frac{Y_j}{Q/N} \right \rceil $ with the same probability in its upper quantized transmission delay process. In other words, for every $m \in \mathbb{N}$, we have the following:
    \begin{equation}\small
        \Pr\bigg(Y_j^{upp} = m\frac{Q}{N}  \bigg) = \Pr\Bigg( Y_j \in \bigg((m-1)\frac{Q}{N}, m\frac{Q}{N}\bigg]\Bigg)
    \end{equation}
    \item \underline{Lower Quantized Transmission Delay Process:} If a transmission delay $Y_j$ occurs with a probability in a transmission delay process, the transmission delay is quantized to $\frac{Q}{N}\left \lfloor \frac{Y_j}{Q/N} \right \rfloor $ with the same probability in its lower quantized transmission delay process. In other words, for every $m \in \mathbb{N}$, we have the following:
    \begin{equation}\small
        \Pr\bigg(Y_j^{low} = (m-1)\frac{Q}{N}  \bigg) = \Pr\Bigg( Y_j \in \bigg[(m-1)\frac{Q}{N}, m\frac{Q}{N}\bigg)\Bigg) \\
    \end{equation}
\end{enumerate}
\par Even though the transmission delays are quantized, an optimal policy can determine waiting periods in the real interval $[0, M]$. Hence, the state space is still an uncountable set. The next proposition allows us to restrict the state space to a finite set.
\begin{prop}\label{CountState}
When a quantization on the transmission delay is performed for any number of sub-intervals $N$, there exists an optimal update policy whose request points are in the set $\Big\{ 0, \frac{Q}{N}, \frac{2Q}{N}, \dots, Q \Big\}$. 
\end{prop}
\begin{proof}
The proof is provided in Appendix \ref{app:CountState}.
\end{proof}
\par The state and action spaces for lower and upper quantizations of a transmission delay process becomes finite because the ages at the delivery points are quantized and the possible delivery points form a finite set as a result of Proposition \ref{CountState}. Then, we can define the spaces of $A_j, Y_j,$ and $Z_j$ as follows:
\begin{deftn}
For a given query $Q$, let us define the following sets:
\begin{itemize}
    \item $\mathcal{A}^N \delequal  \bigg\{0, \frac{Q}{N}, \frac{2Q}{N}, \dots, Q \bigg\}$
    \item $\mathcal{Z}^N \delequal \bigg\{ 0, \frac{Q}{N}, \frac{2Q}{N}, \dots, \frac{\left \lfloor \frac{M}{Q/N} \right \rfloor Q}{N} \bigg\}$
    \item $\mathcal{Y}^N \delequal \bigg\{ \frac{\left \lfloor \frac{B_L}{Q/N} \right \rfloor Q}{N}, \frac{\big(\left \lfloor \frac{B_L}{Q/N} \right \rfloor+1 \big) Q}{N}, \dots, \frac{\left \lceil \frac{B_U}{Q/N} \right \rceil Q}{N} \bigg\}$
\end{itemize}
\end{deftn}
Up to now, we have only analyzed the optimal update policy for quantized transmission delays. The next proposition puts an upper and a lower bound to the optimal expected age penalty for an unquantized transmission delay process. Furthermore, it proposes an update policy whose expected age penalty lays between the upper and lower bounds with the help of characterization in Section \ref{sect:SingleQueryCase}. 
\begin{prop}\label{QuantizedProp}
For any given transmission delay process and the number of sub-intervals $N$, the following hold:
\begin{enumerate}[(i)]
    \item There exists an update policy for an unquantized transmission delay process whose expected age penalty is less than or equal to the optimal age penalty for the upper quantized transmission delay process. \label{QuantizedPropi}
    \item The optimal expected age penalty for lower quantization of a transmission delay process is less than or equal to the optimal expected age penalty for the unquantized transmission delay process.\label{QuantizedPropii}
\end{enumerate}
\end{prop}
\begin{proof}
For the proof of (i), we construct an update policy for an unquantized transmission delay process whose expected age penalty is less than or equal to the optimal expected age penalty for the upper quantized transmission delay process. There exists an optimal update policy for the upper quantized transmission delay process by Proposition \ref{SuffandExis}. Let $\pi^{opt}_1$ be a deterministic optimal policy that is characterized in Corollary \ref{CharacOfPolicy}. Let $z^{opt}(.,.)$ be the decision function of the update policy $\pi^{opt}_1$. The constructed optimal policy determines $R_{j+1}$ for $j \geq 0$ as the following:
\begin{equation}\small
    R_{j+1} \delequal Q - \frac{Q}{N}\left \lceil \frac{A_j}{Q/N}\right \rceil - z^{opt}\bigg(\frac{Q}{N}\left \lceil \frac{\Delta(Y_j)}{Q/N}\right \rceil,\frac{Q}{N}\left \lceil \frac{A_j}{Q/N}\right \rceil \bigg)
\end{equation}
\par Now, let us prove that this constructed policy gives the desired expected age penalty. Let $(Y_i)_{i=1}^J$ where $J$ is an arbitrary natural number be a transmission delay sequence from the unquantized transmission delay process when an update packet at the border point is requested. The correspondence of the transmission delay sequence on the upper quantized transmission delay process is $\bigg(\frac{Q}{N}\left \lceil \frac{Y_i}{Q/N}\right \rceil \bigg)_{i=1}^J$. If the constructed policy follows the steps above, then the request points are the same for $(Y_i)_{i=1}^J$ and $\bigg(\frac{Q}{N}\left \lceil \frac{Y_i}{Q/N}\right \rceil\bigg)_{i=1}^J$. Thus, for any $D_j$ where $1 \leq j \leq J$, the AoI in the interval $\bigg[Q-D_j, \frac{Q}{N}\left \lceil \frac{(Q-D_j)}{Q/N}\right \rceil \bigg)$ is smaller for the unquantized transmission delay process. For every point outside this interval, the AoI will be the same for both of the transmission delay processes. This is valid for every transmission delay sequence $(Y_i)_{i=1}^J$, hence the expected age penalty for the unquantized transmission delay process is less than or equal to the optimal expected age penalty for the upper quantized transmission delay process.
\par The proof of (ii) is similar to the previous part. Let $\pi^{opt}_1$ be a deterministic optimal policy for the unquantized transmission delay process. By Proposition \ref{SuffandExis}, $\pi^{opt}_1$ can find the optimal waiting period for every $\Delta(D_j)$ and $Q-D_j$. If the destination nodes follow the same update policy $\pi^{opt}_1$ for the lower quantized transmission delay process, the obtained expected age penalty is less than or equal to the optimal expected age penalty for the unquantized transmission delay process. This completes the proof.
\end{proof}
\par The optimal update policy for the upper quantization of a transmission delay process enables us to construct an update policy for the transmission delay process. The expected age penalty resulting from this constructed update policy is proved to lays between the optimal expected age penalties of the upper and lower quantized transmission delay processes. Furthermore, we show in the next proposition that the upper and lower bounds converge to each other as $N$ increases. Thus, we can find an update policy whose expected age penalty is arbitrarily close to the optimal expected age penalty for any transmission delay process and age penalty function. 
\begin{prop}\label{UpperLowerApproach}
For $\epsilon >0$, there exists $N_1 \in \mathbb{N}$ such that the difference between optimal expected age penalties of upper and lower quantizated transmission delay processes is less than $\epsilon$ if the quantization is performed with $N \geq N_1$ sub-intervals.  
\end{prop}
\begin{proof}
The proof is provided in Appendix \ref{app:UpperLowerApproach}.
\end{proof}
\par Propositions \ref{QuantizedProp} and \ref{UpperLowerApproach} employ optimal solutions of the upper and lower quantized transmission delay processes while constructing an update policy for the unquantized transmission delay process. Hence, the remaining part of this subsection is to solve the stochastic shortest path problem for quantized transmission delay processes. When the transmission delay process is quantized, the problem turns out to be a stochastic shortest path problem with finite state and action spaces as a result of Proposition \ref{CountState}. This problem class can be solved by the value iteration method given the explicit cost of each action in each state\cite{Bertsakas_2005}.
\par To provide an explicit cost of each action in each state, we again use the function $G^{\pi}_R$. We prove in Proposition \ref{SuffandExis} that there exists a deterministic policy $\pi^{opt}_1=z(Y_j, A_j)$ that is optimal for a given transmission delay process $\{Y_i\}$. Then, the following can be obtained by incorporating \eqref{GD} into \eqref{GR}:
 \begin{equation}\small \label{GRonly}
     \begin{split}
         &G^{\pi^{opt}_1}_R\bigg(Q - R_j, \Delta(R_j)\bigg) \\
    & = E\bigg[ G^{\pi^{opt}_1}_R(Q-R_j-Y_j-Z_j, Y_j+Z_j)  
   \bigg| Y_j+Z_j \leq Q-R_j \bigg] \\ 
   &\hspace{3.2cm}\times \Pr(Y_j + Z_j \leq Q-R_j) \\
    &+g\big(Q-R_j+\Delta(R_j)\big) \times \Pr(Y_j + Z_j > Q-R_j) 
     \end{split}
 \end{equation}
 where $Z_j = z(Y_j, A_j)$.
\begin{algorithm}[ht]
 \caption{Solution of the Single Query Problem}\label{algorithm}
 \begin{algorithmic}[1]
  \STATE \textbf{given} tolerance $\epsilon$ and sufficiently large $N$ 
  \STATE \textbf{repeat}
  \FOR {$i = 1$ to $\text{length}(\mathcal{A}^N)$}
  \FOR {$j = 1$ to $\text{length}(\mathcal{Y}^N)$}
  \FOR {$k = 1$ to $\text{length}(\mathcal{Z}^N)$}
  \STATE Calculate 
  \begin{equation}\small
      \begin{split}
          &{}^{upper}G^{\pi^{opt}}_R\bigg(\mathcal{A}^N(i) - \mathcal{Z}^N(k), \mathcal{Y}^N(j) + \mathcal{Z}^N(k)\bigg), \\
          &{}^{lower}G^{\pi^{opt}}_R\bigg(\mathcal{A}^N(i) - \mathcal{Z}^N(k), \mathcal{Y}^N(j) + \mathcal{Z}^N(k)\bigg) \nonumber
      \end{split}
  \end{equation}
  by using \eqref{GRonly}
  \ENDFOR
  \STATE 
  \begin{equation}\small
      \begin{split}
          &{}^{upper}z\bigg(\mathcal{Y}^N(j), \mathcal{A}^N(i)\bigg) = \\
          & \max \Bigg\{ \argmin_{x \in \mathcal{Z}^N}  {}^{upper}G^{\pi^{opt}}_R\bigg(\mathcal{A}^N(i) - x, \mathcal{Y}^N(j) + x\bigg) \Bigg\}\nonumber
      \end{split}
  \end{equation}
  \STATE 
  \begin{equation}\small
      \begin{split}
          &{}^{lower}z\bigg(\mathcal{Y}^N(j), \mathcal{A}^N(i)\bigg) = \\
          & \max \Bigg\{ \argmin_{x \in \mathcal{Z}^N}  {}^{lower}G^{\pi^{opt}}_R\bigg(\mathcal{A}^N(i) - x, \mathcal{Y}^N(j) + x\bigg) \Bigg\}\nonumber
      \end{split}
  \end{equation}
  \ENDFOR
  \ENDFOR
  \STATE  $N = 2N$
  \STATE \textbf{until} 
  {\small ${}^{upper}G^{\pi^{opt}}_R(Q) - {}^{lower}G^{\pi^{opt}}_R(Q) < \epsilon$}
 \RETURN {\small ${}^{upper}z$}  
 \end{algorithmic}
 \end{algorithm}
 
\par The single query problem is explicitly solved in Algorithm \ref{algorithm}. In this algorithm, the functions ${}^{upper}G^{\pi^{opt}}_R$ and ${}^{lower}G^{\pi^{opt}}_R$ denote the expected age penalties for the upper and lower quantized transmission delays, respectively. These functions are recursively calculated by using \eqref{GRonly} similar to the value iteration method. This calculation is performed through the loop in $\mathcal{A}^N$ with ascending order. The optimal waiting time for a pair $(Y_j, A_j) \in \mathcal{Y}^N \times \mathcal{A}^N$ is determined by minimizing the function $G_R^{\pi}$ in the set $\mathcal{Z}^N$. Note that the set $\mathcal{A}^N, \mathcal{Y}^N,$ and $\mathcal{Z}^N$ is employed in the algorithm as if they are arrays. 
\par The output of Algorithm \ref{algorithm} is a decision function of $Y_j$ and $A_j$ that characterizes an optimal update policy of the single query problem for the upper quantized transmission delay process. An optimal policy of the single query problem for the unquantized transmission delay process is constructed by an optimal update policy for the upper quantized transmission delay process as it is shown in Proposition \ref{QuantizedProp}(\ref{QuantizedPropi}). Then, the constructed update policy is applied to each consecutive query interval, which is optimal for the PoW problem under periodic query arrivals.  

\section{PoW Problem Dominates UoW Problem} \label{sect:SuperioritySection}
\par Different from the UoW problem, designed as a push-based communication model, we define the PoW problem, designed as a pull-based communication model. The pull-based communication model has an extra knowledge of when the destination node utilizes the upcoming packets. Thanks to this extra knowledge, we are motivated in this study to find an update policy whose query average age penalty is less than or equal to the optimal time average age penalty with the same power constraint. Until this section, we provide a method to achieve the optimal query average age penalty. However, there is no close form expression to compare the optimal query average age penalty with the optimal time average age penalty.
\par In this section, for every i.i.d. transmission delay process, and periodic query arrival process; we construct an update policy whose query average age penalty is less than or equal to the optimal time average age penalty that is found in \cite{updateorwait}. Furthermore, the constructed policy satisfies the same power constraint. Let us redefine the set of stationary and deterministic policies similar to the definition of them in \cite{updateorwait}, and state the theorem, which is the main idea of this section.
\begin{deftn}
\begin{itemize}
    \item []
    \item A policy $\pi \in \Pi$ is said to be a stationary and deterministic policy with the function of $Y_j$ if there exists a decision function $z^{UoW}:[0, \infty) \rightarrow [0, M]$ such that $Z_j = z^{UoW}(Y_j)$ for all $j = 0,1, \dots$.
    \item The set of all stationary and deterministic policies with the function of $Y_j$ is denoted as $\Pi^{UoW}_{SD}$.
\end{itemize}
\end{deftn}
\begin{theorem}\label{superiorityTheorem}
If the transmission delay process $\{Y_j\}_{j=0}^{\infty}$ is i.i.d. such that $\Pr(Y_j \in [B_L, B_U] ) = 1$ and the query arrival process $\{Q_k\}_{k=1}^{\infty}$ is deterministic and periodic, then $\bar{h}_{opt} \leq \bar{g}_{opt}$ with the same power constraint for every period $T$.
\end{theorem}
\begin{proof}
It was shown in \cite{updateorwait} that there exists an optimal policy $\pi^{opt} \in \Pi^{UoW}_{SD}$ for the UoW problem. We construct an update policy that determines waiting periods as follows: $Z_0=z^{UoW}(Y_0) + x$ and $Z_j = z^{UoW}(Y_j), j \geq 1$ where $z^{UoW}$ is the decision function for the policy $\pi^{opt}$. We will show in Appendix \ref{app:superiorityTheorem} that for every i.i.d. transmission delay process, the first observation of transmission delay $Y_0$, and the query period $T$, there exists $x \in [0, T]$ such that the constructed update policy achieves a better or equal query average age penalty than the optimal time average age penalty. As the constructed optimal policy modifies only the first waiting time, the same power constraint is satisfied. The details are in Appendix \ref{app:superiorityTheorem}.
\end{proof} 
\vspace{-0.3cm}
\section{Numerical Results}\label{sect:NumericalResults}
\par Throughout the section, we exhibit the behavior of the average age penalties for the PoW and UoW problems under different transmission delay processes. To be consistent with our system model which assumes finite valued transmission delay, we will utilize truncated versions of certain transmission delay distributions such as exponential and log-normal distributions. Specifically, we truncate the values to start at $0.01$ and go up to a maximum value chosen such that the cumulative distribution of the transmission delay at this value is $0.95$. We choose $T=4B_U$.
\begin{table}[ht]
\renewcommand{\arraystretch}{1.2}
 		\centering
 	\resizebox{.45\textwidth}{!}{%
	\begin{tabular}
		{|c|*{7}{>{\centering\arraybackslash}p{1cm}|}}
		\hline 
		\centering 
		\diagbox{{$\frac{Q}{N}$}}{{$\lambda$}} & $1$ &$1.2$ &$1.4$ &$1.6$ &$1.8$ &$2$  \\
		\cline{1-7} 
		\centering 
$0.16$&$1.297$&$1.097$&$0.897$&$0.792$&$0.702$&$0.624$\\ \cline{1-7} $0.08$&$1.359$&$1.159$&$0.958$&$0.854$&$0.763$&$0.684$\\ \cline{1-7} $0.04$&$1.391$&$1.191$&$0.99$&$0.885$&$0.795$&$0.715$\\ \cline{1-7} $0.02$&$1.407$&$1.207$&$1.006$&$0.901$&$0.811$&$0.731$\\ \cline{1-7} 
		\hline 
	\end{tabular}}
	\caption{Lower bounds on the query average age with i.i.d. truncated exponential distributed service times} \label{LowerBoundAgePenalty}
	\renewcommand{\arraystretch}{1.2}
 		\centering
 	\resizebox{.45\textwidth}{!}{%
	\begin{tabular}
		{|c|*{7}{>{\centering\arraybackslash}p{1cm}|}}
		\hline 
		\diagbox{{$\frac{Q}{N}$}}{{$\lambda$}} &$1$ &$1.2$ &$1.4$ &$1.6$ &$1.8$ &$2$  \\
		\cline{1-7} 
$0.16$&$1.457$&$1.257$&$1.057$&$0.952$&$0.862$&$0.784$\\ \cline{1-7} $0.08$&$1.439$&$1.239$&$1.038$&$0.934$&$0.843$&$0.764$\\ \cline{1-7} $0.04$&$1.431$&$1.231$&$1.03$&$0.925$&$0.835$&$0.755$\\ \cline{1-7} $0.02$&$1.427$&$1.227$&$1.026$&$0.921$&$0.831$&$0.751$\\ \cline{1-7} 
		\hline 
	\end{tabular}}
	\caption{Upper bounds on the query average age with i.i.d. truncated exponential distributed service times} \label{UpperBoundAgePenalty}
\renewcommand{\arraystretch}{1.2}
 		\centering
 	\resizebox{.45\textwidth}{!}{%
	\begin{tabular}
		{|c|*{7}{>{\centering\arraybackslash}p{1cm}|}}
		\hline 
		\diagbox{{$\frac{Q}{N}$}}{{$\lambda$}} &$1$ &$1.2$ &$1.4$ &$1.6$ &$1.8$ &$2$  \\
		\cline{1-7} 
$0.16$&$7\times 10^4$&$5\times 10^4$&$3\times 10^4$&$2\times 10^4$&$2\times 10^4$&$1\times 10^4$\\ \cline{1-7} $0.08$&$6\times 10^5$&$4\times 10^5$&$2\times 10^5$&$2\times 10^5$&$1\times 10^5$&$1\times 10^5$\\ \cline{1-7} $0.04$&$5\times 10^6$&$3\times 10^6$&$2\times 10^6$&$1\times 10^6$&$1\times 10^6$&$8\times 10^5$\\ \cline{1-7} $0.02$&$4\times 10^7$&$3\times 10^7$&$1\times 10^7$&$1\times 10^7$&$7\times 10^6$&$6\times 10^6$\\ \cline{1-7}
		\hline 
	\end{tabular}}
	\caption{Number of $G^{\pi}_R$ calculations to determine an optimal update policy when service times are i.i.d. truncated exponential distribution and the penalty function is identity.} \label{NumberOfCalculations}
\end{table}
\par We compare three different update policies: the zero-wait policy, the optimal policy of the UoW problem found in \cite{updateorwait}, and the optimal policy of the PoW problem found in Algorithm \ref{algorithm}. The optimal solutions of the UoW problem and the PoW problem are referred to as UoW-optimal policy and PoW-optimal policy, respectively. The average age penalty of the PoW-optimal policy is calculated by averaging the age penalties at the query instants. The average age penalties of the zero-wait policy and UoW-optimal policy are calculated as time average age penalties. Perhaps surprisingly, in all of our simulations, the time-average AoI and QAoI are identical for the zero-wait and UoW-optimal policies. The reason is, in all of our examples $X_j = Y_j + Z_j$ obeys the \enquote{Case 1 i.i.d.} random variable definition in \cite{robbins_1953}. Case 1 random variables are all the random variables except the cases that there exists $\beta \in \mathbb{R}$ such that $\Pr(X_j \in \{ k \beta : \ k \in \mathbb{N}\}) = 1$ or $\Pr(X_j = 0) = 1$. The proof for the equivalence of the time-average AoI and QAoI is subject to our future works.\footnote{A related discussion is provided in the third scenario of Appendix \ref{app:ExistenceOfLimitQueryAverageProblem}.} Note that the random variable $X_j$ under the PoW-optimal policy may not be an i.i.d. random variable, that is why the PoW-optimal policy can result in a lower age than the time-average age of the UoW-optimal policy.

\par Tables \ref{LowerBoundAgePenalty}, \ref{UpperBoundAgePenalty}, and \ref{NumberOfCalculations} illustrate the change in lower bounds of the query average age, upper bounds of the query average age, and the number of calculations to find an optimal policy, respectively for different numbers of sub-intervals $N$ under i.i.d. truncated exponentially distributed services times. Observing the tables \ref{LowerBoundAgePenalty} and \ref{UpperBoundAgePenalty}, we detect that the upper bound is much stricter than the lower bound. This is also the case for the other transmission delay processes such as log-normal, beta, uniform distributions. Even though the number of calculations is exponentially increasing as the number of sub-intervals $N$ increases, the upper bounds of the query average age are rapidly converging. It means that reaching a satisfactory approximate solution for the PoW problem does not require an excessive number of $G^{\pi}_R$ calculations. As a result, we decide to present only the upper bound of the query average age to avoid confusion in the following figures. 
\begin{figure}[ht] 
\centering
\psfrag{Optimal Policy of UoW Problem}[ll][ll][0.8]{Optimal Policy of UoW Problem}
\psfrag{Upper Bound on}[ll][ll][0.8]{Upper Bound on}
\psfrag{Optimal Policy of PoW Problem}[ll][ll][0.8]{Optimal Policy of PoW Problem}
\psfrag{Zero-Wait Policy}[ll][ll][0.8]{Zero-Wait Policy}
\psfrag{Average Age at Query Instances}[cc][cc][0.9]{Average Age at Query Instants}
\psfrag{alpha}[cc][cc]{$\alpha = \beta$}
\psfrag{1}[ll][ll][0.9]{1}
\psfrag{0.95}[ll][ll][0.9]{0.95}
\psfrag{0.9}[ll][ll][0.9]{0.9}
\psfrag{0.85}[ll][ll][0.9]{0.85}
\psfrag{0.8}[ll][ll][0.9]{0.8}
\psfrag{0.75}[ll][ll][0.9]{0.75}
\psfrag{0.7}[ll][ll][0.9]{0.7}
\psfrag{0.65}[ll][ll][0.9]{0.65}
\psfrag{0.6}[ll][ll][0.9]{0.6}
\psfrag{0.4}[ll][ll][0.9]{0.4}
\psfrag{0.2}[ll][ll][0.9]{0.2}
\psfrag{0}[ll][ll][0.9]{0}
\includegraphics [scale = 1]{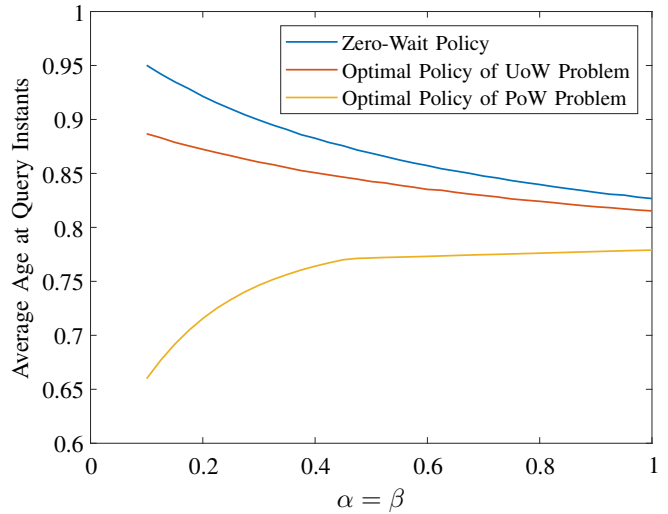}
\caption{Average age at query instants with i.i.d. beta distributed services times. The optimal policies of the PoW problems are found with $\frac{Q}{N} = 0.05$.}

\label{fig:BetaDistribution}
\end{figure}

\begin{figure}[ht] 
\centering
\psfrag{Optimal Policy of UoW Problem}[ll][ll][0.8]{Optimal Policy of UoW Problem}
\psfrag{Upper Bound on}[ll][ll][0.8]{Upper Bound on}
\psfrag{Optimal Policy of PoW Problem}[ll][ll][0.8]{Optimal Policy of PoW Problem}
\psfrag{Zero-Wait Policy}[ll][ll][0.8]{Zero-Wait Policy}
\psfrag{Average Age at Query Instances}[cc][cc][0.9]{Average Age at Query Instants}
\psfrag{sigma}[cc][cc]{$\sigma$}
\psfrag{0.8}[ll][ll][0.9]{0.8}
\psfrag{0.9}[ll][ll][0.9]{0.9}
\psfrag{1}[ll][ll][0.9]{1}
\psfrag{1.1}[ll][ll][0.9]{1.1}
\psfrag{1.2}[ll][ll][0.9]{1.2}
\psfrag{1.3}[ll][ll][0.9]{1.3}
\psfrag{1.4}[ll][ll][0.9]{1.4}
\psfrag{1.5}[ll][ll][0.9]{1.5}
\psfrag{1.6}[ll][ll][0.9]{1.6}
\psfrag{1.7}[ll][ll][0.9]{1.7}

\psfrag{5.5}[ll][ll][0.9]{5.5}
\psfrag{5}[ll][ll][0.9]{5}
\psfrag{4.5}[ll][ll][0.9]{4.5}
\psfrag{4}[ll][ll][0.9]{4}
\psfrag{3.5}[ll][ll][0.9]{3.5}
\psfrag{3}[ll][ll][0.9]{3}
\psfrag{2.5}[ll][ll][0.9]{2.5}
\psfrag{2}[ll][ll][0.9]{2}
\psfrag{1.5}[ll][ll][0.9]{1.5}

\includegraphics [scale = 1]{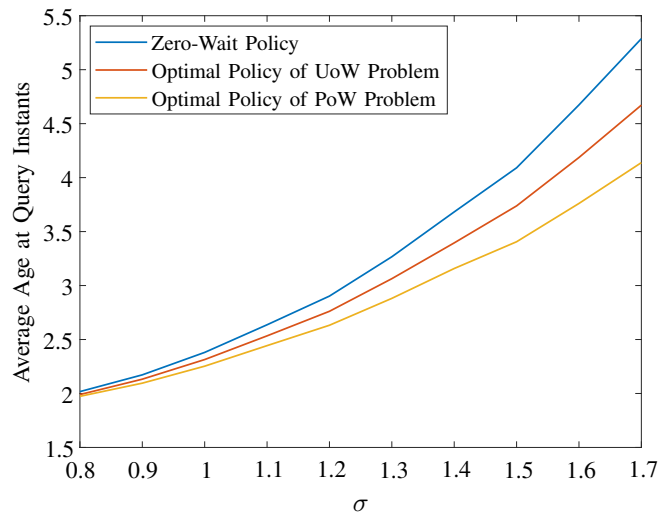}
\caption{Average age at query instants with i.i.d. truncated log-normal distributed services times with parameters $(\sigma, \mu)$ where $\mu=0$. The optimal policies of the PoW problems are found with $\frac{Q}{N} = 0.2$.}

\label{fig:LogNormalDistribution}
\end{figure}

\begin{figure}[ht] 
\centering
\psfrag{Optimal Policy of UoW Problem}[ll][ll][0.8]{Optimal Policy of UoW Problem}
\psfrag{Upper Bound on}[ll][ll][0.8]{Upper Bound on}
\psfrag{Optimal Policy of PoW Problem}[ll][ll][0.8]{Optimal Policy of PoW Problem}
\psfrag{Average Age at Query Instances}[cc][cc][0.9]{Average Age at Query Instants}
\psfrag{alpha}[cc][cc]{$\alpha$}
\psfrag{1.2}[ll][ll][0.9]{1.2}
\psfrag{1.4}[ll][ll][0.9]{1.4}
\psfrag{1.6}[ll][ll][0.9]{1.6}
\psfrag{1.8}[ll][ll][0.9]{1.8}
\psfrag{2}[ll][ll][0.9]{2}
\psfrag{2.2}[ll][ll][0.9]{2.2}

\psfrag{3}[ll][ll][0.9]{3}
\psfrag{4}[ll][ll][0.9]{4}
\psfrag{5}[ll][ll][0.9]{5}
\psfrag{6}[ll][ll][0.9]{6}
\psfrag{7}[ll][ll][0.9]{7}
\psfrag{8}[ll][ll][0.9]{8}
\psfrag{9}[ll][ll][0.9]{9}
\psfrag{10}[ll][ll][0.9]{10}

\includegraphics [scale = 1]{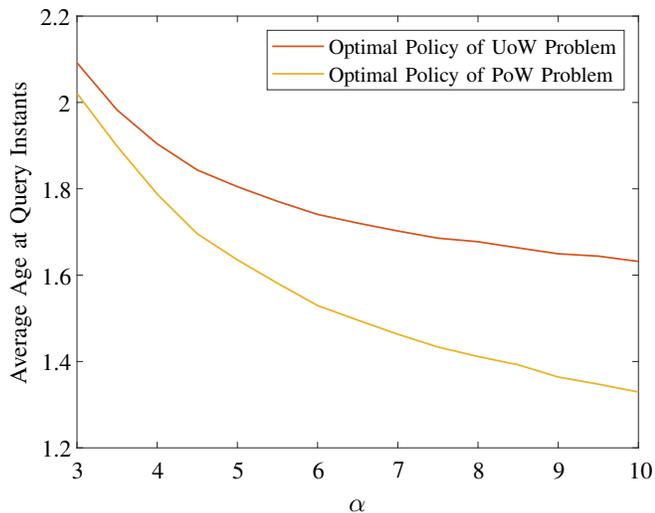}
\caption{Average age at query instants with i.i.d. truncated Pareto distributed services times with $(x_m, \alpha)$ where $x_m = 1$. The optimal policies of the PoW problems are found with $\frac{Q}{N} = 0.05$. Note that the optimal policy of the UoW problem is equivalent to the zero wait policy when $x_m=1$ and $\alpha \geq 3$. }
\label{fig:ParetoDistribution}
\end{figure}

\begin{figure}[ht] 
\centering
\psfrag{Optimal Policy of UoW Problem}[ll][ll][0.8]{Optimal Policy of UoW Problem}
\psfrag{Upper Bound on}[ll][ll][0.8]{Upper Bound on}
\psfrag{Optimal Policy of PoW Problem}[ll][ll][0.8]{Optimal Policy of PoW Problem}
\psfrag{Zero-Wait Policy}[ll][ll][0.8]{Zero-Wait Policy}
\psfrag{Average Age Penalty at Query Instances}[cc][cc][0.9]{Average Age Penalty at Query Instants}
\psfrag{alpha}[cc][cc]{$\alpha$}
\psfrag{2}[ll][ll][0.9]{2}
\psfrag{1}[ll][ll][0.9]{1}
\psfrag{1.2}[ll][ll][0.9]{1.2}
\psfrag{1.4}[ll][ll][0.9]{1.4}
\psfrag{1.6}[ll][ll][0.9]{1.6}
\psfrag{1.8}[ll][ll][0.9]{1.8}

\psfrag{0}[ll][ll][0.9]{0}
\psfrag{50}[ll][ll][0.9]{50}
\psfrag{100}[ll][ll][0.9]{100}
\psfrag{150}[ll][ll][0.9]{150}
\psfrag{200}[ll][ll][0.9]{200}
\psfrag{250}[ll][ll][0.9]{250}

\includegraphics [scale = 1]{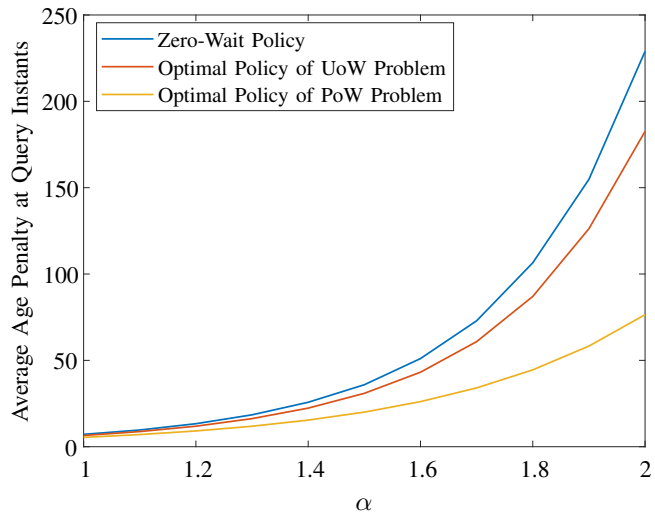}
\caption{Average age penalty at query instants with the age penalty function $g(x) = e^{\alpha x}-1$ and i.i.d. truncated exponential distributed services times where $\lambda = 1$. The optimal policies of the PoW problems are found with $\frac{Q}{N} = 0.05$.}

\label{fig:NonLinearAgePenalty}
\end{figure}

\begin{figure}[ht] 
\centering
\psfrag{Optimal Policy of UoW Problem}[ll][ll][0.8]{Optimal Policy of UoW Problem}
\psfrag{with Transmission Constraint}[ll][ll][0.8]{with Transmission Constraint}
\psfrag{Upper Bound on}[ll][ll][0.8]{Upper Bound on}
\psfrag{Optimal Policy of PoW Problem}[ll][ll][0.8]{Optimal Policy of PoW Problem}
\psfrag{Average Age at Query Instances}[cc][cc][0.9]{Average Age at Query Instants}
\psfrag{lambda}[cc][cc]{$\lambda$}
\psfrag{6}[ll][ll][0.9]{6}
\psfrag{5}[ll][ll][0.9]{5}
\psfrag{4}[ll][ll][0.9]{4}
\psfrag{3}[ll][ll][0.9]{3}
\psfrag{2}[ll][ll][0.9]{2}
\psfrag{1.5}[ll][ll][0.9]{1.5}
\psfrag{1}[ll][ll][0.9]{1}
\psfrag{0.5}[ll][ll][0.9]{0.5}
\psfrag{0}[ll][ll][0.9]{0}

\includegraphics [scale = 1]{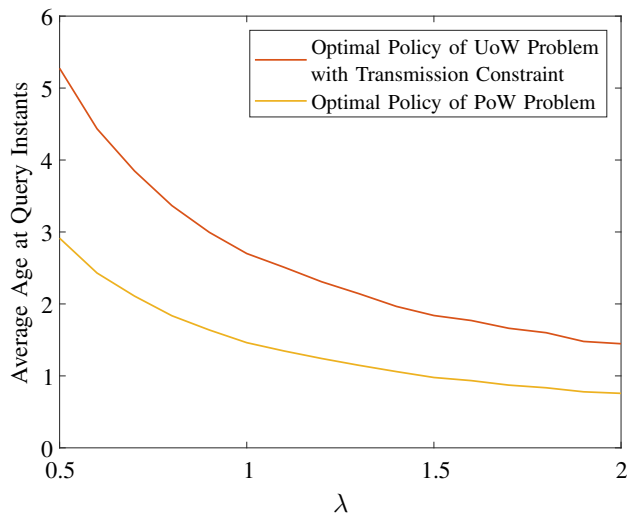}
\caption{Average age at query instants with i.i.d. truncated exponential distributed services times with the parameter $\lambda$ when the optimal policy of the UoW problem is constrained to transmit the same number of update packets as the optimal policy of PoW the problem. The optimal policies of the PoW problems are found with $\frac{Q}{N} = 0.05$.}

\label{fig:ExponentialDistributionConstraint}
\end{figure}

\begin{figure}[ht] 
\centering
\psfrag{Optimal Policy of UoW Problem}[ll][ll][0.8]{Optimal Policy of UoW Problem}
\psfrag{with Transmission Constraint}[ll][ll][0.8]{with Transmission Constraint}
\psfrag{Upper Bound on}[ll][ll][0.8]{Upper Bound on}
\psfrag{Optimal Policy of PoW Problem}[ll][ll][0.8]{Optimal Policy of PoW Problem}
\psfrag{Average Age at Query Instances}[cc][cc][0.9]{Average Age at Query Instants}
\psfrag{alpha}[cc][cc]{$\alpha$}
\psfrag{1}[ll][ll][0.9]{1}
\psfrag{1.5}[ll][ll][0.9]{1.5}
\psfrag{2}[ll][ll][0.9]{2}
\psfrag{2.5}[ll][ll][0.9]{2.5}
\psfrag{3.5}[ll][ll][0.9]{3.5}
\psfrag{4.5}[ll][ll][0.9]{4.5}

\psfrag{3}[ll][ll][0.9]{3}
\psfrag{4}[ll][ll][0.9]{4}
\psfrag{5}[ll][ll][0.9]{5}
\psfrag{6}[ll][ll][0.9]{6}
\psfrag{7}[ll][ll][0.9]{7}
\psfrag{8}[ll][ll][0.9]{8}
\psfrag{9}[ll][ll][0.9]{9}
\psfrag{10}[ll][ll][0.9]{10}

\includegraphics [scale = 1]{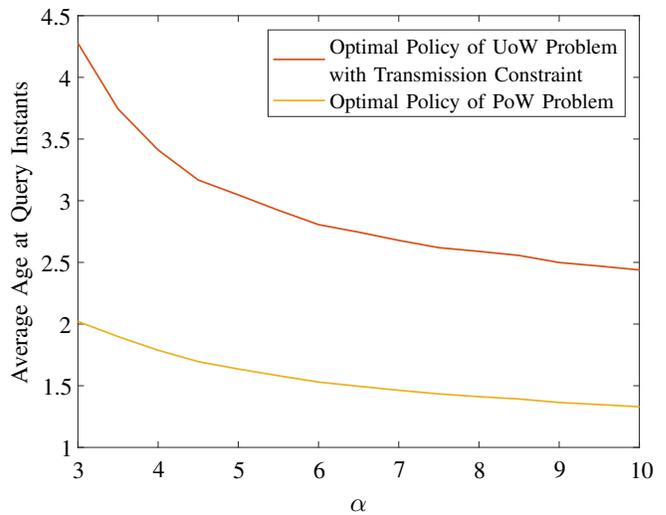}
\caption{Average age at query instants with i.i.d. truncated Pareto distributed services times with the parameters $(x_m, \alpha)$ where $x_m=1$, when the optimal policy of the UoW problem is constrained to transmit the same number of update packets as the optimal policy of the PoW problem. The optimal policies of the PoW problems are found with $\frac{Q}{N} = 0.05$.}

\label{fig:ParetoDistributionConstraint}
\end{figure}

\par Figures \ref{fig:BetaDistribution}, \ref{fig:LogNormalDistribution}, and \ref{fig:ParetoDistribution} illustrate the behavior of the average ages under i.i.d. beta distributed service times with equal $\alpha, \beta$ parameters, i.i.d. truncated log-normal distributed service times, i.i.d. truncated Pareto distributed service times, respectively. When $\alpha = \beta = 1$, the Beta distribution becomes a uniform distribution between $0$ and $1$. As $\alpha = \beta$ approaches $0$, it approaches a bimodal distribution concentrated around $0$ and $1$ with probability close to $0.5$ each. Interestingly, as $\alpha$ and $\beta$ increase, the average ages of the zero-wait policy and the UoW-optimal policy decrease whereas the average age of the PoW-optimal policy increases even though the mean of the beta distribution is constant, $\frac{\alpha}{\alpha + \beta} = \frac{1}{2}$. The benefit of using the PoW-optimal policy is pronounced when the transmission delay is bi-modal distributed. The log-normal distribution is a heavy-tailed distribution especially for large $\sigma$. We observe in Figure \ref{fig:LogNormalDistribution} that the PoW-optimal policy performs better than the other policies in heavy-tailed distribution as well. On the other hand, as $\alpha$ goes to $\infty$, the Pareto distribution converges to the dirac delta function $\delta(t-x_m)$, similar to the example 1 in Section \ref{sect:introduction}. We choose $x_m = 1$ which leads that UoW-optimal policy is equivalent to the zero wait policy for $\alpha \geq 3$ \cite[Theorem~5]{updateorwait}. We observe in Figure \ref{fig:ParetoDistribution} that PoW-optimal policy performs well as the transmission delay distribution approaches the dirac delta function.

\par Figure \ref{fig:NonLinearAgePenalty} exhibits the behavior of the average age penalties for different $\alpha$ when the age penalty function $g(x) = e^{\alpha x}-1$ and service times are exponentially distributed with $\lambda = 1$. This nonlinear age penalty function represents destination nodes that demand very fresh update packets and harshly penalize stale update packets. In the figure, we observe that the PoW-optimal policy works much better than the other policies especially for high $\alpha$ values. It means that the pull-based communication model is beneficial to utilize when the destination node demands very fresh update packets. 

\par Up to now, we have not put any constraint on the number of transmissions for the policies. Figures \ref{fig:ExponentialDistributionConstraint} and \ref{fig:ParetoDistributionConstraint} illustrate the behavior of the average ages under truncated i.i.d. exponential distributed service times and Pareto distributed service times, respectively, when the number of transmissions in the UoW-optimal policy is constrained by the number of transmissions made by the PoW-optimal policy. We observe that the average age of the PoW-optimal policy is much lower than the average age of the UoW-optimal policy for an equal number of transmissions. This implies that in a practical situation, applying the PoW solution can be significantly more energy-efficient, for the same age performance.

\vspace{-0.3cm}
\section{Conclusions and Future Directions }\label{sect:conclusion}
\par We studied the optimal control of the status update system in which the destination node requests the source node to submit an update packet to the channel. We defined a continuous, non-decreasing, and non-negative penalty function to represent the level of dissatisfaction on data staleness. While solving the PoW problem, we first identified the PoW problem under the single query case as a stochastic shortest path problem with uncountable state and action spaces. For this specific SSP problem, we obtained an optimal policy. Using the solution of the SSP problem, we found out an optimal policy for the PoW problem under periodic query arrival processes. Furthermore, we provided an analytical comparison between the UoW and PoW problems: (i) An optimal policy that minimizes the UoW problem also minimizes the PoW problem under Poisson query arrivals. Furthermore, their average age penalties are equivalent. (ii) The optimal query average age penalty under periodic query arrivals is always less than or equal to the optimal time average age penalty. An interesting by product is that for a large class of distributions, the QAoI achieved by Zero-Wait and the UoW-optimal policies are identical to the time-average AoI achieved by these policies, and both are remarkably higher than the QAoI achieved by the PoW-optimal policy, even when the former two are allowed an unconstrained number of transmissions. For the same number of tranmissions, the PoW-optimal result achieves a more significant lowering of QAoI, which in turn implies the potential energy efficiency of a PoW-optimal solution for a desired Query AoI performance.

Future directions for this work include the general solution of the PoW problem (i.e., for general query arrival processes, and delay processes with memory), and exhibiting the superiority of the result to those obtained by previous push-based solutions.
\appendices
\section{Proof of Proposition \ref{Poisson}}\label{app:ProofOfPoisson}
\par Let $\Pi^{UoW}_{Pois}, \Pi^{PoW}_{Pois}$ be the sets of optimal causal policies for the UoW and PoW problems, respectively, for a given transmission delay process under a Poisson query arrival process. We prove in this proof that $\Pi^{UoW}_{Pois} \subset \Pi^{PoW}_{Pois}$ and $\Pi^{PoW}_{Pois} \subset \Pi^{UoW}_{Pois}$ for every transmission delay process, which completes the first part of the proposition. 
\par Let $N_t$ be a Poisson counting process with a parameter $\lambda$. Then, $N_t$ has the stationary and independent increments property. By Taylor expansion, we can state that
\begin{equation}
    \Pr(N_{t+\delta} - N_t = 1) = \lambda \delta e^{-\lambda \delta} = \lambda \delta + o(\delta)
\end{equation}
\par Let us divide the time interval $[0, Q_n]$ into small interval with length $\delta$. Let $\mathbb{P}$ be the partition that consists of these small intervals. Then, an upper Darboux sum can be derived as follows:
\begin{equation}\small
    \begin{split}
    &E\bigg[\sum_{k=1}^{n} g(\Delta(Q_k))\bigg]  \\
    & \leq E\bigg[\sum_{j=0}^{Q_n/\delta}P(N_{t+\delta}-N_t = 1) \times \sup_{t \colon t \in [0, \delta)} g\big(\Delta(j\delta + t)\big) \bigg] \\
    &= E\bigg[\sum_{j=0}^{Q_n/\delta}(\lambda \delta + o(\delta)) \times \sup_{t \colon t \in [0, \delta)} g\big(\Delta(j\delta + t)\big) \bigg] \\
    &\delequal U(g(\Delta), \mathbb{P})
    \end{split}
\end{equation}
Similar to the upper Darboux sum, a lower Darboux sum can be derived as follows:
\begin{equation}\small
\begin{split}
    &E\bigg[\sum_{k=1}^{n} g(\Delta(Q_k))\bigg] \\
    & \geq  E\bigg[\sum_{j=0}^{Q_n/\delta}P(N_{t+\delta}-N_t = 1) \times \inf_{t \colon t \in [0, \delta)} g\big(\Delta(j\delta + t)\big) \bigg] \\
    &= E\bigg[\sum_{j=0}^{Q_n/\delta}(\lambda \delta + o(\delta)) \times \inf_{t \colon t \in [0, \delta)} g\big(\Delta(j\delta + t)\big) \bigg] \\
    & \delequal L(g(\Delta), \mathbb{P})
\end{split}
\end{equation}
\par Let $\mathbb{P}_1$ and $\mathbb{P}_2$ be the partitions consisted of small intervals with length $\delta$ in which there is no delivery and there is a delivery, respectively. This means that $\mathbb{P} = \mathbb{P}_1 \cup \mathbb{P}_2$ and $\mathbb{P}_1 \cap \mathbb{P}_2 = \emptyset$. Then, there exists $\delta_1 > 0$ such that the partition $\mathbb{P}_1$ is organized with $\delta_1$ length small intervals and   $U(g(\Delta), \mathbb{P}_1)-L(g(\Delta), \mathbb{P}_1 ) < \epsilon/2$ since $g(\Delta(.))$ is continuous on the partition $\mathbb{P}_1$. Inside any interval in $\mathbb{P}_2$, the supremum point is less than or equal to $g(B_U + M)$ while the infimum point is greater than or equal to $g(B_L)$. On the other hand, the number of small intervals in $\mathbb{P}_2$ can be at most $Q_n/B_L$. Therefore if the partition $\mathbb{P}_2$ is organized with small intervals with length $\delta_2$ equal to $\frac{\epsilon \times B_L}{3Q_n\times(g(M+B_U)-g(B_L))}$, then $U(g(\Delta), \mathbb{P}_2)-L(g(\Delta), \mathbb{P}_2) \leq \epsilon/3 < \epsilon/2$. As a result, if the partition $\mathbb{P}$ is organized with $\delta = \min(\delta_1, \delta_2)$ length small intervals, then $U(g(\Delta), \mathbb{P}) - L(g(\Delta), \mathbb{P}) < \epsilon$ because $U(g(\Delta), \mathbb{P}_1) + U(g(\Delta), \mathbb{P}_2) = U(g(\Delta), \mathbb{P})$ and $L(g(\Delta), \mathbb{P}_1) + L(g(\Delta), \mathbb{P}_2) = L(g(\Delta), \mathbb{P})$. Hence, for every $n \in \mathbb{N}$, we have proved the following by \cite[Theorem~6.6]{rudin_1976}:
\begin{equation}\label{PoissonIntegration}
    E\bigg[\sum_{k=1}^{n} g(\Delta(Q_k))\bigg] = E\bigg[\lambda \int^{Q_n}_0 g(\Delta( t)) dt\bigg] 
\end{equation}
Note that $\lambda$ is just a constant and \eqref{PoissonIntegration} holds for every $n \in \mathbb{N}$. Then, we can obtain the following:
\begin{equation}\label{PoissonIntegration1}
    \begin{split}
        \limsup_{n \rightarrow \infty} \frac{ E\bigg[\mathlarger{\sum_{k=1}^{n}} g(\Delta(Q_k))\bigg]}{n} 
        &\stackrel{(a)}{=}\limsup_{n \rightarrow \infty} E\Bigg[\frac{ \int^{Q_n}_0 g(\Delta( t)) dt}{Q_n}\Bigg] \\ 
        &\stackrel{(b)}{=} \limsup_{n \rightarrow \infty} E\Bigg[\frac{ \int^{D_n}_0 g(\Delta( t)) dt}{D_n}\Bigg]
    \end{split}
\end{equation}
where (a) follows from $E[Q_n] = n/\lambda$. (b) can be shown by using the following two facts: (i) Let $D_i$ be the closest delivery point to a query $Q_k$. Then, $|D_i-Q_k| < (B_U + M)/2$. (ii) the function $g(\Delta)$ has upper and lower bounds. As a result of \eqref{PoissonIntegration1}, minimizing \eqref{Problem2} and \eqref{Problem1} are equivalent, which implies that $\Pi^{UoW}_{Pois}$ and $\Pi^{PoW}_{Pois}$ are equivalent. Furthermore, their average age penalties are equivalent by \eqref{PoissonIntegration1}.

\section{Proof of Proposition \ref{SuffandExis}} \label{app:SuffandExis}
We first prove that $Q-D_j$ and $\Delta(D_j)$ are sufficient statistics to obtain an optimal $Z_j$ for every $j$, $(Y_i)_{i=0}^{j}$, and $(Z_i)_{i=0}^{j-1}$. We perform induction on $Q-D_j$. Let us map each $Q-D_j$ to a natural number $n$ such that $(n-1)B_L \leq Q-D_j < nB_L$. If $n=1$, then $Q-D_j < B_L$. For every waiting period $Z_j$, the age penalty at the query is constant because a new update cannot arrive until the query. Then, the age penalty at the query is $g(Q-D_j+\Delta(D_j))$. Thus, if $n = 1$, $Q-D_j$ and $\Delta(D_j)$ are sufficient statistics to obtain an optimal $Z_j$ for every $j$, $(Y_i)_{i=0}^{j}$, and $(Z_i)_{i=0}^{j-1}$. Let us assume that $Q-D_j$ and $\Delta(D_j)$ are sufficient statistics to obtain an optimal $Z_j$ for $n = 2,3, \dots K$ where $K$ is an arbitrary natural number. Let $\Pi_K$ be the set of all causal waiting policies such that if $\pi \in \Pi_K$; then $\pi$ determines waiting times at delivery points $D_j: Q-D_j < KB_L$ solely based on $Q-D_j$ and $\Delta(D_j)$, for the delivery points $D_j : Q-D_j \geq KB_L$, the waiting policy may not determine the waiting time based on $Q-D_j$ and $\Delta(D_j)$. Due to the induction assumption, the single query problem can be minimized in the set of $\Pi_K$. Let us prove that the single query problem can be minimized in the set of $\Pi_{K+1}$ as well. For every $\pi \in \Pi_K$, we can obtain the following:
\begin{equation}\small \label{penaltyterm} 
\begin{split}
    &G_R^{\pi}\bigg(Q-D_j-Z_j, \Delta(D_j)+Z_j, (Y_i)_{i=0}^{j},  (Z_i)_{i=0}^{j}\bigg) \\
    &\stackrel{(a)}{=} E\bigg[ G^{\pi}_D\bigg(Q-D_j-Z_j-Y_{j+1}, Y_{j+1}, (Y_i)_{i=0}^{j+1},  (Z_i)_{i=0}^{j}\bigg)  \\
    & \hspace{0.75cm} \bigg| Y_{j+1} + Z_j \leq Q-D_j \bigg]  \times \Pr\bigg(Y_{j+1} + Z_j \leq Q-D_j\bigg)  \\
    & \hspace{0.3cm} +g\bigg(Q-D_j +\Delta(D_j)\bigg) \times \Pr\bigg(Y_{j+1} + Z_j> Q-D_j\bigg)
\end{split}
\end{equation}
where (a) follows from \eqref{GR}. $Q-D_{j+1} = Q-D_j-Z_j-Y_{j+1} < K B_L$ as $Y_{j+1} \geq B_L$ and $Q-D_j < (K+1)B_L$. This means that we can exploit the induction assumption in the RHS of \eqref{penaltyterm} to claim that $(Y_i)_{i=0}^{j}$ and $(Z_i)_{i=0}^{j-1}$ does not affect the value of the term with expectation given $Q-D_{j+1} = Q-D_j-Y_{j+1}-Z_j$ and $\Delta(D_{j+1})=Y_{j+1}$. This is because $\pi \in \Pi_K$. In the term with penalty function, only $Q-D_j$ and $\Delta(D_j)$ appear. This means that the optimal control problem of choosing an optimal $Z_j$ at the delivery point $D_j$ does not depend on $(Y_i)_{i=0}^{j}$ and $(Z_i)_{i=0}^{j-1}$. This completes the induction. Once the single query problem can be minimized in the set of $\bigcup_{K=1}^{\infty}\Pi_{K}$, it is easy to show that the calculation of the functions $G^{\pi}_D$ and $G^{\pi}_R$ can be performed by only knowing $Q-D_j$ and $\Delta(D_j)$ for every $\pi \in \bigcup_{K=1}^{\infty}\Pi_K$. The proof can be performed with a similar induction.   
\par From now on, we can omit $(Y_i)_{i=0}^{j}$ and $(Z_i)_{i=0}^{j-1}$ from $G^{\pi}_R$ and $G^{\pi}_D$. For the part related to the existence of a deterministic optimal policy, we construct a deterministic optimal policy by performing another induction on $Q-D_j$. Before move on to the induction, we state some simple observation.
\begin{lemma}\label{AppendixBLemma}
Let us assume that there exists a deterministic optimal policy $\pi_1^{opt}$.
\begin{enumerate}[(i)]
    \item \label{LemmaB1}Let $h \colon \mathbb{R} \rightarrow \mathbb{R}$ such that $h(\epsilon) \delequal max_{x \in [0, M+B_U]}g(x+\epsilon)-g(x)$. Then, we can obtain the following for every $t_1, t_2 \in \mathbb{R}$
    \begin{equation}
        0 \leq G_R^{\pi_1^{opt}}(t_1, t_2 + \epsilon) - G_R^{\pi_1^{opt}}(t_1, t_2) \leq h(\epsilon)
    \end{equation}
    \item \label{LemmaB2}If $f(x) \delequal G^{\pi_1^{opt}}_R(Q-D_j-x, \Delta(D_j))$ is a lower semi-continuous function for a given $Q-D_j$ and $\Delta(D_j)$, then $f'(x) \delequal G^{\pi_1^{opt}}_R(Q-D_j-x, \Delta(D_j)+x)$ is a lower semi-continuous function as well.
    \item \label{LemmaB3}If $f(x) \delequal G^{\pi_1^{opt}}_R(Q-D_j-x, \Delta(D_j))$ is a lower semi-continuous function for every $Q-D_j$ and $\Delta(D_j)$ satisfying $Q-D_j<C$, where $C$ is an arbitrary real number, then $f''(x) \delequal G^{\pi_1^{opt}}(Q-D_i-Y_{i+1}-x, Y_{i+1})$ is a lower semi-continuous function as well for every $Q-D_i$ and $\Delta(D_i)$ satisfying $Q-D_i<C$.
    \item \label{LemmaB4}For every $\epsilon > 0$, there exists $\delta > 0$ such that
    \begin{equation}
        \Pr(t_1 < Y_j \leq t_1 + \delta) < \epsilon 
    \end{equation}
    where $t_1$ is a given real number satisfying $t_1 \in [B_L, B_U]$.
    \item \label{LemmaB5}For every $Q-D_j$, $\Delta(D_j)$, and $\epsilon > 0$, there exists $\delta > 0$ such that 
    \begin{equation}
        \begin{split}
            G^{\pi_1^{opt}}_D(Q-D_j-&Z_j-Y_{j+1}, Y_{j+1}) - \\ 
            &g(Q-D_j + \Delta(D_j)-\delta) < \epsilon
        \end{split}
    \end{equation}
\end{enumerate}
\end{lemma}
\begin{proof}
\begin{enumerate}[(i)]
    \item It follows from \eqref{GR} and the facts that the penalty function $g$ is continuous and non-decreasing.
    \item It follows from the definition of lower semi-continuity and Lemma \ref{AppendixBLemma}(\ref{LemmaB1})
    \item It follows from Lemma \ref{AppendixBLemma}(\ref{LemmaB2}) and the fact that $\pi_1^{opt}$ is a deterministic optimal policy. 
    \item The transmission delay is measurable on Borel algebra on the real line.
    \item It follows from \eqref{GD} and the fact that the penalty function $g$ is continuous and non-decreasing.
\end{enumerate}
\end{proof}
\par The idea which will be proven by the induction is that  $f(x) \delequal G^{\pi_1^{opt}}_R(Q-D_j-x, \Delta(D_j))$ is a lower semi-continuous function. From Lemma \ref{AppendixBLemma}(\ref{LemmaB2}), $f'(x) \delequal G^{\pi_1^{opt}}_R(Q-D_j-x, \Delta(D_j)+x)$ is a lower semi-continuous function as well. Therefore, it attains its infimum for every $Q-D_j$ and $\Delta(D_j)$ due to the extension of Extreme Value Theorem to semi-continuity. Then, this infimum point can be determined as the waiting time at the delivery point $D_j$. This policy is a deterministic optimal policy that decides the waiting periods solely based on $Q-D_j$ and $\Delta(D_j)$. Now, let us move on to the induction. When $Q-D_j < B_L$, then all waiting periods result in the same age penalty. This means that there exists a deterministic optimal policy $\pi_1^{opt}$ for $n=1$. Additionally, $f(x) \delequal G^{\pi_1^{opt}}_R(Q-D_j-x, \Delta(D_j))$ is lower semi-continuous for every $D_j$ satisfying $Q-D_j < B_L$. Let us assume for $n=2$ that $f(x) \delequal G^{\pi_1^{opt}}_R(Q-D_j-x, \Delta(D_j))$ is lower semi-continuous for every $D_j$ satisfying $B_L \leq Q-D_j < 2B_L$. Note that the superscript $\pi_1^{opt}$ refers in the definition of the function $f$ that the deterministic optimal policy $\pi_1^{opt}$ is performed starting with $(j+2)^{th}$ request because $(j+1)^{th}$ request has already determined as $D_j+x$. The delivery point $D_{j+1}$ must satisfy $Q-D_{j+1} < (2-1)B_L$ in which there exists a deterministic optimal policy. As $f(x) \delequal G^{\pi_1^{opt}}_R(Q-D_j-x, \Delta(D_j))$ is lower semi-continuous for $n=2$, $f'(x) \delequal G^{\pi_1^{opt}}_R(Q-D_j-x, \Delta(D_j)+x)$ is a lower semi-continuous function as well by Lemma \ref{AppendixBLemma}(\ref{LemmaB2}). Hence, the function $f'$ attains its minimum for every $B_L \leq Q-D_j < 2B_L$ and $\Delta(D_j)$. Therefore, there exists a deterministic optimal policy for $n=2$ as well. Similar to the transition from $n=1$ to $n=2$, let us assume one by one that the function $f(x) \delequal G^{\pi_1^{opt}}_R(Q-D_j-x, \Delta(D_j))$ is lower semi-continuous and there exists a deterministic optimal policy for $n=2,3,\dots, K$ where $K$ is an arbitrary natural number. Let us prove that $f(x) \delequal G^{\pi_1^{opt}}_R(Q-D_j-x, \Delta(D_j))$ is a lower semi-continuous function for $n=K+1$. To reach contradiction, suppose that the claim is false. Then, there exists $Q-D_j, \Delta(D_j),$ and $x_0$ satisfying $KB_L \leq Q-D_j < (K+1)B_L$ such that $f(x) \delequal G^{\pi_1^{opt}}_R(Q-D_j-x, \Delta(D_j))$ is not lower semi-continuous at $x_0$. Hence, there exist either an increasing or a decreasing sequence $(x_n)$ and $C > 0$ such that $\lim_{n \rightarrow \infty} x_n = x_0$ and $f(x_n) - f(x_0) < -C$ for every $n \in \mathbb{N}$. 
\par \underline{If $(x_n)$ is a increasing sequence}, then we obtain the following by \eqref{GR}:
\begin{equation}\label{SubstractionOfGR}\small
\begin{split}
    f(x_n) - &f(x_0) = A\times \Pr(Y_{j+1} \leq Q-D_j-x_0)   \\
    &+ B\times \Pr(Q-D_j-x_0 < Y_{j+1} \leq Q-D_j-x_n) \\
    &+C\times \Pr(Q-D_j-x_n < Y_{j+1})
\end{split}
\end{equation}
where $A,B,$ and $C$ are the following:
\begin{equation}\small
    \begin{split}
        A =& E\bigg[G^{\pi_1^{opt}}_D\big(Q-D_j-x_n-Y_{j+1}, Y_{j+1}\big) \\ 
        &- G^{\pi_1^{opt}}_D\big(Q-D_j-x_0-Y_{j+1}, Y_{j+1}\big) \bigg| Y_{j+1} \leq Q-D_j-x_0\bigg]
    \end{split}
\end{equation}
\begin{equation}\small
    \begin{split}
        B =& g(Q-D_j+\Delta(D_j)) - E\bigg[G^{\pi_1^{opt}}_D\big(Q-D_j-x_n \\
        &-Y_{j+1}, Y_{j+1}\big) \bigg| Q-D_j-X_n < Y_{j+1} \leq Q-D_j-x_0\bigg]\hspace{0.7cm}
    \end{split}
\end{equation}
\begin{equation}\small
    C = g(Q-D_j-x_n+\Delta(D_j)) - g(Q-D_j-x_0+\Delta(D_j))
\end{equation}
\par From the induction assumption and Lemma \ref{AppendixBLemma}(\ref{LemmaB3}), $A$ can be arbitrarily small. $B$ is upper bounded by $g(M+B_U)$ and the multipliers of $B$ in \eqref{SubstractionOfGR} can be arbitrarily small by Lemma \ref{AppendixBLemma}(\ref{LemmaB4}). $C$ can be arbitrarily small due to the continuity of the penalty function $g$. Therefore, there exists $x_n$ such that $f(x_n)-f(x) \geq -C$, which is a contradiction. 
\par \underline{If $(x_n)$ is a decreasing sequence}, then an equation similar to \eqref{SubstractionOfGR} can be written. The terms that are similar to $A$ and $C$ can be analyzed similarly. The term that is similar $B$ can be analyzed with the help of Lemma \ref{AppendixBLemma}(\ref{LemmaB5}). After the analysis, a similar contradiction can be achieved. 
\par As a result, the function $f(x)$ is lower semi-continuous. From Lemma \ref{AppendixBLemma}(\ref{LemmaB2}), the function $f'(x)$ is lower semi-continuous for every $Q-D_j$ and $\Delta(D_j)$ satisfying $KB_L \leq Q-D_j < (K+1)B_L$. Thus, the function attains its infimum, and the infimum point can be determined as a deterministic optimal waiting period $Z_j$, which completes the induction. 
\section{Proof of Proposition \ref{BorderPointProp}} \label{app:BorderPointProp}
\par We start this proof with a lemma:
\begin{lemma}\label{RequestUntilQ-3BU}
For any delivery point $D_j \in [0, Q-2B_U]$ and its AoI $\Delta(D_j)$, an optimal request point $R_{j+1}$ must be until $Q-B_U$ \ie\, $R_{j+1} \leq Q-B_U$.
\end{lemma}
\begin{proof}
Let us assume that this lemma is not true: There exist a delivery point $D_j \in [0, Q-2B_U]$ and its AoI at the delivery $\Delta(D_j)$ such that an optimal request point is $R_{j+1}> Q-B_U$. Let this policy follows $\pi^{opt}$ and let $R^* \delequal R_{j+1}$. We will show that there exists $\pi^{modified}$ such that $G^{\pi^{modified}}_D(Q-D_j, \Delta(D_j)) \leq G^{\pi^{opt}}_D(Q-D_j, \Delta(D_j))$. Let $\pi^{modified}$ determine ${R}^{mod}_{j+1} = D_j$ and ${R}^{mod}_{j+2} = R^*$. As the time duration between $D_j$ and $R^*$ is greater than $B_U$, $\pi^{modified}$ can determine ${R}^{mod}_{j+2}$ as $R^*$ regardless of the transmission delay of the $(j+1)^{th}$ update. After the request at $R^*$, let $\pi^{modified}$ imitate $\pi^{opt}$. This means that $G^{\pi^{opt}}_R(Q-t_1, t_2) = G^{\pi^{modified}}_R(Q-t_1, t_2)$ for every $t_1 \geq R^*$ and $t_2 \in [B_L, B_U+M]$. As a result of the modification, we can state that
\begin{equation}\label{modifiedGD}
    \begin{split}
         G^{\pi^{modified}}_D&\bigg(Q-D_j, \Delta(D_j)\bigg) \\
         &\stackrel{(a)}{=} G^{\pi^{modified}}_R\bigg(Q-{R}^{mod}_{j+2}, \Delta(R^{mod}_{j+2})\bigg)\\ &\stackrel{(b)}{=} G^{\pi^{opt}}\bigg(Q-R^*, \Delta({R}^{mod}_{j+2})\bigg)
    \end{split}
\end{equation}
where (a) follows from the decision of $\hat{R}_{j+1}$ and $\hat{R}_{j+2}$, and (b) follows from the fact that $\pi^{modified}$ imitates $\pi^{opt}$ starting from the point $R^*$.
\par On the other hand, as $\pi^{opt}$ determines the request point $R_{j+1}$ as $R^*$, we can state that
\begin{equation}\label{optGD}
    G^{\pi^{opt}}_D\bigg(Q-D_j, \Delta(D_j)\bigg) = G^{\pi^{opt}}_R\bigg(Q-R^*, \Delta(R_{j+1})\bigg)
\end{equation}
As $\Delta(R_{j+1}) > \Delta({R}^{mod}_{j+2})$, and $Q-R^* < B_U$, we can say that $G^{\pi^{opt}}_R(Q-R^*, \Delta(R_{j+1})) > E[G^{\pi^{opt}}(Q-R^*, \Delta({R}^{mod}_{j+2}))]$\footnote{If there exists $x < B_U$ such that $\Pr(Y_j \in (x, B_U]) = 0$, then $B_U$ can be shifted to $x$. Thus, we can assume that $\Pr(Y_j \in (x, B_U]) > 0$ for every $x < B_U$. As a result, we can claim that $G^{\pi^{opt}}_R(Q-R^*, \Delta(R_{j+1}))$ is strictly greater than $E[G^{\pi^{opt}}(Q-R^*, \Delta({R}^{mod}_{j+2}))]$.} by \eqref{GR}. As a result of \eqref{modifiedGD} and \eqref{optGD}, $G^{\pi^{modified}}_D(Q-D_j, \Delta(D_j)) < G^{\pi^{opt}}_D(Q-D_j, \Delta(D_j))$ that contradicts with the fact that $\pi^{opt}$ is the optimal policy. Hence, there is no such $D_j$, which completes the proof.
\end{proof}
As a result of Lemma \ref{RequestUntilQ-3BU}, $R_{j+1}^* \leq Q-B_U$. From Proposition \ref{AoISubNotAffect}, AoI at $R_{j+1}^*$ does not affect the age penalty at the query. Next, we prove that there is no $R_{j+1} \in [0, Q-B_U]$ such that $G^{\pi_1^{opt}}_R(Q-R_{j+1}^*) > G^{\pi_1^{opt}}_R(Q-R_{j+1})$. If there existed such $R_{j+1} \in [Q-3B_U, Q-B_U]$, then the destination node would determine the optimal request point for the delivery point $D_j^*$ as $R_{j+1}$. Hence, we can state the following for every delivery point $D_j$ and its transmission delay $Y_j$ satisfying $D_j \in [Q-3B_U, Q-2B_U]$ and $Y_j \in [B_L, B_U]$:
\begin{equation}\label{BorderPointGRandGD}
    G_R^{\pi^{opt}}(Q-R_{j+1}^*) \leq G_D^{\pi^{opt}}(Q-D_j, Y_j)
\end{equation} 
\par On the other hand, such $R_{j+1}$ cannot be in the interval $[0, Q-3B_U]$ as well. This statement is proved by induction. Similar to the proof of Proposition \ref{SuffandExis}, $R_{j+1}$ is mapped to a natural number $n$ if it satisfies $(n-1)B_L \leq Q-3B_U-R_{j+1} < nB_L$. It is true for $n=1$ \ie\, such $R_{j+1}$ cannot be in the interval $0 \leq Q-3B_U-R_{j+1} < B_L$ because of the following:
\begin{equation}\small
    \begin{split}
        G^{\pi_1^{opt}}_R\big(Q-R_{j+1}\big) &\stackrel{(a)}{=} E\bigg[G^{\pi_1^{opt}}_D\big(Q-R_{j+1}-Y_{j+1}, Y_{j+1}\big)\bigg] \\
        & \stackrel{(b)}{\geq} G^{\pi_1^{opt}}_R\big(Q-R_{j+1}^*\big)
    \end{split}
\end{equation}
where (a) follows from \eqref{GR}, and (b) follows from \eqref{BorderPointGRandGD}. 
 Let us assume that the induction statement is true for $n= 2, 3,\dots, K$ where $K$ is an arbitrary natural number. This statement assumes the following for every request point $R_{j+1}$ satisfying $0 \leq Q-3B_U-R_{j+1} < K B_L$:
 \begin{equation}\label{BorderPointInductionAssumption}
     G_R^{\pi^{opt}}(Q-R_{j^*+1}) \leq G_R^{\pi^{opt}}(Q-R_{j+1})
 \end{equation}
 Let us prove the induction statement for $n = K+1$. For every $R_{j+1}$ satisfying $K B_L \leq Q-3B_U-R_{j+1} < (K+1) B_L$, we have the following:
 \begin{equation}\small
     \begin{split}
         G^{\pi_1^{opt}}_R\big(Q-R_{j+1}\big) &\stackrel{(a)}{=} E\bigg[G^{\pi_1^{opt}}_D\big(Q-R_{j+1}-Y_{j+1}, Y_{j+1}\big)\bigg] \\
         &\stackrel{(b)}{\geq} G^{\pi_1^{opt}}_R\big(Q-R_{j+1}^*\big)
     \end{split}
 \end{equation}
where (a) follows from \eqref{GR}, and (b) follows from \eqref{BorderPointInductionAssumption}. This implies that the induction is completed. 
\par As a result, for every $R_{j+1} \in [0, Q-B_U]$, we have the following:
\begin{equation}
    G^{\pi_1^{opt}}_R(Q-R_{j+1}^*) \leq G^{\pi_1^{opt}}_R(Q-R_{j+1})
\end{equation}
It means that $R_{j+1}^*$ attains its infimum value on the interval $[0, Q-B_U]$. This completes the proof.

\section{Proof of Proposition \ref{CountState}} \label{app:CountState}
We perform a similar induction included in the proof of Proposition \ref{SuffandExis}. Let us map each $Q-D_j$ to a natural number $n$ that satisfies $(n-1)B_L \leq Q-D_j < nB_L$. If $n=1$, the request point does not affect the expected age penalty at the query. Thus, requesting at the query is an optimal request point that proves the proposition statement for $n=1$. Let us assume that the optimal request point is in the set $\{ 0, \frac{N}{Q}, \frac{2N}{Q}, \dots, Q \}$ when a delivery occurs at time $D_j$ satisfying $(n-1)B_L \leq Q-D_j < nB_L$ for $n= 1, 2, \dots, K$ where $K$ is an arbitrary natural number. Let us prove that the optimal request point is in the set $\{ 0, \frac{N}{Q}, \frac{2N}{Q}, \dots, Q \}$ when a delivery occurs at time $D_j$ satisfying $KB_L \leq R-D_j < (K+1)B_L$. Let us assume the inverse. There exists a delivery point $D_j$ such that $KB_L \leq Q-D_j < (K+1)B_L$ and the is no optimal request point in the set $\{ 0, \frac{N}{Q}, \frac{2N}{Q}, \dots, Q \}$. As there exists an optimal policy from Proposition \ref{SuffandExis}, there exists an optimal request point $R_{j+1} \not\in \{ 0, \frac{N}{Q}, \frac{2N}{Q}, \dots, Q \}$. For every quantized transmission delay $Y_{j+1}$, the next delivery point satisfies $Q-D_{j+1} < KB_L$. If $Q-D_{j+1} > 0$, the optimal next request point should be in the set $\{ 0, \frac{N}{Q}, \frac{2N}{Q}, \dots, Q \}$ due to the induction assumption. Instead of requesting at $R_{j+1}$, if the request was performed at $R_{j+1}^{mod} \delequal \frac{\left \lceil \frac{R_{j+1}}{N/Q} \right \rceil N}{Q}$, there would be two cases based on the transmission delay $Y_{j+1}$. For every $m \in \mathbb{N}$; if $D^{mod}_{j+1} > m \frac{Q}{N}$, then $D_{j+1} > m \frac{Q}{N}$; if $D^{mod}_{j+1} < m \frac{Q}{N}$, then $D_{j+1} < m \frac{Q}{N}$ because of the quantized transmission delay process, where $m \frac{Q}{N}$ represents the possible next request point or the query. Therefore, requesting an update packet at $R_{j+1}^{mod}$ is optimal given that $R_{j+1}$ is an optimal request point. This conclusion contradicts with the assumption. Hence, there exists an optimal request point in the set $\{ 0, \frac{N}{Q}, \frac{2N}{Q}, \dots, Q \}$ for every delivery point, which completes the proof.

\section{Proof of Proposition \ref{UpperLowerApproach}} \label{app:UpperLowerApproach}
\par Let $\delta = \frac{1}{2} \max_{x \in [0, M+B_U]} g(x+\epsilon) - g(x)$. Let $N_1$ be a natural number satisfying $N_1 > \frac{3 Q^2}{4 \delta B_L}$. There exists a deterministic optimal policy $\pi_1^{opt}$ whose first request point is the border point for lower quantization of the transmission delay with $N\geq N_1$ by Corollary \ref{CharacOfPolicy}. We construct an update policy $\pi_1^{mod}$ for upper quantization of the transmission delay with $N$ by utilizing $\pi_1^{opt}$. Let the border point corresponding to the lower quantized transmission delay and $\pi_1^{opt}$ be $Q^{BP}$. Let $\pi_1^{mod}$ pull its first request at $R_1^{mod} \delequal Q^{BP}-\frac{Q}{N}\times \frac{Q-Q^{BP}}{B_L}$. Note that $R_1^{mod}-Q^{BP} < \delta$ as $Q^{BP} \geq Q-3B_U  > \frac{Q}{4}$ and $N \geq N_1$. Let $(Y_j)_{j=1}^J$ be an arbitrary transmission delay sequence from the unquantized transmission delay process where $\sum_{j=1}^J Y_j > Q$. Let $(Y_j^{upp})_{j=1}^J$ and $(Y_j^{low})_{j=1}^J$ be the sequences that correspond to upper and lower quantized of $(Y_j)_{j=1}^J$, respectively. Let $(Z_j)_{j=1}^J$ be the waiting time sequences that is causally determined by $\pi_1^{opt}$ based on $(Y_j^{low})_{j=1}^J$. If $\pi_1^{mod}$ determines the waiting periods the same as $\pi_1^{opt}$ after the first request point \ie\, $R_{j+1}^{mod} = R_j^{mod} + Y_j^{upp} + Z_j, j \geq 1$, then the difference between age penalties under $\pi_1^{opt}$ and $\pi_1^{mod}$ is less than $\epsilon$. This is because $ 0 < R_j - R_j^{mod} < \delta$ for every $j$ where $R_j$ is the $j^{th}$ request point under $\pi_1^{opt}$. As this is valid for every transmission delay sequence $(Y_j)_{j=1}^J$, its expected difference is less than $\epsilon$. This completes the proof.
\section{Proof of Theorem \ref{superiorityTheorem}} \label{app:superiorityTheorem}
The objective of the UoW problem in \eqref{Problem1} attains its limit under any policy $\pi \in \Pi^{UoW}_{SD}$ from \cite[Eq~13]{updateorwait}. Therefore, we can define $g_{opt, Y_0}$ with a limit operation to measure the optimal time average age penalty when the first realization of transmission delay is $Y_0$.
\begin{equation}\label{newGFunc}
    g_{opt, Y_0} = \min_{\pi \in \Pi_{SD}^{UoW}} \lim\limits_{n\rightarrow \infty} \frac{E_{\mathcal{Y}}\big[\int^{D_n}_0 g(\Delta( t)) dt \big| Y_0\big]}{E[D_n]}
\end{equation}
where the expectation is taken with respect to transmission delay sequences. Let us define a function $f \colon ([0, T], \Pi_{SD}^{UoW}) \rightarrow \Pi$ in an effort to generate a constructed update policy.
\begin{align} \label{Piout}
    f(x, \pi^{in}) = \pi^{out}, \qquad x \in [0,T]
\end{align}
where $\pi^{in}$ is a stationary deterministic policy in which $Z_j = z(Y_j)$, $z : [0, \infty) \rightarrow [0,M]$. Then $\pi^{out}$ is a causal policy in which $Z_j = z(Y_j)$ with the same z function for $i > 0$ and $Z_0 = z(Y_0) + x$. Note that $\pi^{out}$ is stationary and deterministic policy, which is a function of $Y_j$ for every $Z_j$ for $j > 0$. $j = 0$ does not satisfy stationary and deterministic property. In Appendix \ref{app:ExistenceOfLimitQueryAverageProblem}, it is proved that the limit of the PoW problem exists when the performed update policy is a stationary and deterministic policy, which is a function of $Y_j$ for every $Z_j$ for $j > 0$. Therefore, we can define a function $h \colon (\mathbb{R}^+, \Pi_{SD}) \rightarrow \mathbb{R}^+$ with limit operation in order to measure query average age penalty when the first realization of transmission delay is $Y_0$ and the performed update policy is $\pi$.
\begin{equation}\label{hDef}
    h(Y_0, \pi) = \lim\limits_{n\rightarrow \infty} \frac {E_{\mathcal{Y}}\big[ \sum^n_{k=1} g(\Delta (Q_k)) \big| Y_0\big]}{n}
\end{equation}
where the expectation is taken with respect to transmission delay sequences.
\par Let $m$ be a natural number such that $(m+1)T > M+T+B_U$. Let $\Delta_{\pi, \mathbf{Y}}(t)$ denote AoI at time $t$, when a stationary and deterministic policy $\pi$ is performed on a transmission delay sequence $\mathbf{Y} = (Y_0, Y_1, \dots)$. When the performed policy $\pi$ is a stationary and deterministic policy which is function of $Y_j$, $\Delta_{\pi_\mathbf{Y}}(t)$ is a function of $t$. Then, it is obvious from the definition of the function $f$ in \eqref{Piout} that
\begin{equation} \label{TimeShiftDelta1}
    \Delta_{\pi^{opt}, \mathbf{Y}}(t) = \Delta_{f(x, \pi^{opt}), \mathbf{Y}}(t+x), \quad \ for \ t > M+B_U
\end{equation}
\par As \eqref{TimeShiftDelta1} holds for every transmission delay sequence, we take expectation on transmission delay sequences. Thus, the following equation holds for every $t>M+B_U$:
\begin{equation} \label{TimeShiftDelta}
    E_{\mathcal{Y}}\bigg[\Delta_{\pi^{opt}, \mathbf{Y}}(t)\bigg] = E_{\mathcal{Y}}\bigg[\Delta_{f(x, \pi^{opt}), \mathbf{Y}}(t+x)\bigg]
\end{equation}
\par Now, we are going to prove that for every starting point of $Y_0 \in [B_L, B_U]$, there exists $x \in [0,T]$ such that $h(Y_0, f(x, \pi^{opt})) \leq g_{opt, Y_0}$ where $\pi^{opt}$ the optimal update policy for the UoW problem. Suppose that this is not true. Then, there exists $Y_0 \in [B_L, B_U]$ such that for all $x \in [0,T]$, $h(Y_0, f(x, \pi^{opt})) > g_{opt, Y_0}$. Furthermore, $g(\Delta_{\pi^{opt}, \mathbf{Y}}(t))$ is lower semi-continuous with respect to $t$ because $g$ is continuous and non-decreasing. As $g$ is uniformly continuous on the interval $[B_L, B_U+M]$ and bounded in this compact interval; lower semi-continuity of $E_{\mathcal{Y}}[g(\Delta_{f(x,\pi^{opt}), \mathbf{Y}}(Q_k))]$ with respect to $x$ can be easily shown by its definition.\footnote{A clear discussion is provided in Appendix \ref{app:SuffandExis}} Then, $h(Y_0, f(x, \pi^{opt})$ turns out to be sum of countable lower semi-continuous functions with respect to $x$. Countable sum of lower semi-continuous functions is lower semi-continuous when they are lower bounded \cite[Chapter~2]{Rudin87}. As the variable $x$ is in a compact set $[0,T]$, the function attains $h(Y_0, f(x, \pi^{opt})$ its infimum. Hence, there exists $C > 0$ such that 
\begin{equation}\label{ContradictionThrm}
    h(Y_0, f(x, \pi^{opt})) \geq g_{opt, Y_0} + C
\end{equation}
\par Let us rewrite \eqref{newGFunc} in terms of the function $\Delta_{\pi, \mathbf{Y}}(t)$:
\begin{equation}
    g_{opt, Y_0} = \lim\limits_{n\rightarrow \infty} \frac{E_{\mathcal{Y}}\big[\int^{D_n}_0 g(\Delta_{\pi^{opt}, \mathbf{Y}}( t)) dt \big| Y_0\big]}{E[D_n]} \label{SameG}
\end{equation}
\par Let us rewrite \eqref{hDef} in terms of the function $\Delta_{\pi, \mathbf{Y}}(t)$:
\begin{equation}\small \label{DeltaRel}
    \begin{split}
        h\big(Y_0, f(x, \pi^{out})\big) &= \lim\limits_{n\rightarrow \infty} \frac {E_{\mathcal{Y}}\bigg[ \mathlarger{\sum^n_{k=1}} g\big(\Delta_{f(x,\pi^{opt}), \mathbf{Y}} (Q_k)\big) \bigg| Y_0 \bigg]}{n} \\ 
        &\stackrel{(a)}{=} \lim\limits_{n\rightarrow \infty} \frac {E_{\mathcal{Y}}\bigg[ \mathlarger{\sum^n_{k=m+1}} g\big(\Delta_{f(x,\pi^{opt}), \mathbf{Y}} (Q_k)\big) \bigg| Y_0 \bigg]}{n-m} \\
        &\stackrel{(b)}{=}\lim\limits_{n\rightarrow \infty} \frac {E_{\mathcal{Y}}\bigg[ \mathlarger{\sum^n_{k=m+1}} g\big(\Delta_{\pi^{opt}, \mathbf{Y}} (Q_k-x)\big)\bigg| Y_0\bigg]}{n-m} 
    \end{split}
\end{equation} 
where (a) follows from the properties of limit and (b) follows from \eqref{TimeShiftDelta}.
\par Let us divide the interval $[0,T]$ into small intervals with length $\delta$. Then, we obtain the following:
\begin{equation}\small \label{InfSumThrm}
    \begin{split}
        &\frac{\mathlarger{\sum_{i = 1}^{T/\delta}}h\bigg(Y_0, f\big((i-1)\delta, \pi^{opt}\big)\bigg)}{T/\delta} \\
        &\hspace{0.2cm}\stackrel{(a)}{=} \frac{1}{T/\delta} \sum_{i = 1}^{T/\delta} \lim\limits_{n\rightarrow \infty} \frac {E_{\mathcal{Y}}\Bigg[\mathlarger{\sum^n_{k=m+1}} g\bigg(\Delta_{\pi^{opt}, \mathbf{Y}} \big(Q_k-(i-1)\delta\big)\bigg)\Bigg]}{n-m}  \\
        &\hspace{0.2cm} \stackrel{(b)}{=} \lim\limits_{n\rightarrow \infty} \sum_{i = 1}^{T/\delta} E_{\mathcal{Y}}\Bigg[\frac{ \mathlarger{\sum^n_{k=m+1} }g\bigg(\Delta_{\pi^{opt}, \mathbf{Y}} \big(Q_k-(i-1)\delta\big)\bigg)}{\big(n-m\big) T/\delta}\Bigg]\\
        & \hspace{0.2cm} \stackrel{(c)}{=} \lim\limits_{n\rightarrow \infty} E_{\mathcal{Y}}\Bigg[\frac{\mathlarger{\sum_{i = Q_m/\delta + 1}^{Q_n/\delta}} g\bigg(\Delta_{\pi^{opt}, \mathbf{Y}} \big(i\delta\big)\bigg)}{\big(Q_n-Q_m\big)/\delta}\Bigg]
    \end{split}
\end{equation}
where (a) follows from \eqref{DeltaRel}, (b) follows from interchanging the order of the limit and summation by Lebesgue's Dominated Convergence Theorem as all of the terms are upper bounded by $g(B_u+M)$, and (c) follows from exchanging summation and expectation. 
\par As $\delta$ goes to $0$, we obtain the following:
\begin{equation}\small \label{ChangeLimitOrderThrm}
    \begin{split}
        \lim_{\delta \rightarrow 0}&\frac{\mathlarger{\sum_{i = 1}^{T/\delta}}h\bigg(Y_0, f\big((i-1)\delta, \pi^{opt}\big)\bigg)}{T/\delta} \\
        & \hspace{0.0cm} \stackrel{(a)}{=} \lim_{\delta \rightarrow 0} \lim\limits_{n\rightarrow \infty} E_{\mathcal{Y}}\bigg[\frac{\sum_{i = Q_m/\delta + 1}^{Q_n/\delta} g\big(\Delta_{\pi^{opt}, \mathbf{Y}} (i\delta)\big)}{\big(Q_n-Q_m\big)/\delta}\bigg] \\
        &\hspace{0.0cm}\stackrel{(b)}{=} \lim\limits_{n\rightarrow \infty} \lim_{\delta \rightarrow 0}  E_{\mathcal{Y}}\bigg[\frac{\sum_{i = Q_m/\delta + 1}^{Q_n/\delta} g(\Delta_{\pi^{opt}, \mathbf{Y}} (i\delta))}{\big(Q_n-Q_m\big)/\delta}|Y_0 \bigg]\\
        &\hspace{0.0cm} \stackrel{(c)}{=} \lim\limits_{n\rightarrow \infty} E_{\mathcal{Y}}\bigg[ \lim_{\delta \rightarrow 0} \frac{\sum_{i = Q_m/\delta}^{Q_n/\delta} g(\Delta_{\pi^{opt}, \mathbf{Y}} (i\delta))}{\big(Q_n-Q_m\big)/\delta}\bigg|Y_0\bigg] \\
        &\hspace{0.0cm} \stackrel{(d)}{=}\lim\limits_{n\rightarrow \infty} E_{\mathcal{Y}}\bigg[\frac{\int^{Q_n}_{Q_m} g(\Delta_{\pi^{opt}, \mathbf{Y}}( t)) dt}{\big(Q_n-Q_m\big)}\bigg|Y_0 \bigg] \\
        &\hspace{0.0cm} \stackrel{(e)}{=}\lim\limits_{n\rightarrow \infty} E_{\mathcal{Y}}\bigg[\frac{\int^{Q_n}_{0} g(\Delta_{\pi^{opt}, \mathbf{Y}}( t)) dt}{Q_n}\bigg|Y_0 \bigg] \\
        &\hspace{0.0cm} \stackrel{(f)}{=} \lim\limits_{n\rightarrow \infty} E_{\mathcal{Y}}\bigg[\frac{\int^{D_n}_0 g(\Delta_{\pi^{opt}, \mathbf{Y}}( t)) dt}{D_n}\bigg|Y_0 \bigg] \stackrel{(g)}{=} g_{{opt}, Y_0}
    \end{split}
\end{equation}
In \eqref{ChangeLimitOrderThrm}, (a) follows from \eqref{InfSumThrm}. (b) follows from Moore Osgood Theorem as the term with the expectation is proved to be uniformly convergent in Appendix \ref{app:UniformCon}. (c) follows from Lebesgue's Dominated Convergence Theorem as $g(\Delta_{\pi^{opt}, \mathbf{Y}}(.))$ is upper bounded by $g(B_U + M)$. (d) follows from Riemann Integration that is proved in \cite[Theorem~6.10]{rudin_1976}. (e) and (f) are obtained from the following facts: (i) $[0,Q_m]$ is a bounded interval. (ii) $B_L < \Delta_{\pi^{opt}, \mathbb{Y}}(t) < M+B_U$ for all $t\in \mathbb{R}^{+}$. (iii) Let $D_i$ be the closest delivery point to a query $Q_k$. Then $|D_i-Q_k| < (B_U+M)/2$. (g) follows from \eqref{SameG}.
\par On the other hand, the following can be obtained from \eqref{ContradictionThrm}:
\begin{equation}\label{ContradictionThrm1}
    \lim_{\delta \rightarrow 0} \frac{\mathlarger{\sum_{i = 1}^{T/\delta}}h\bigg(Y_0, f\big((i-1)\delta, \pi^{opt}\big)\bigg)}{T/\delta}  \geq g_{opt, Y_0} + C
\end{equation}
\eqref{ContradictionThrm1} contradicts \eqref{ChangeLimitOrderThrm}. Therefore, for every $Y_0 \in [B_L , B_U]$, there exists $x \in [0,T]$ such that $h(Y_0, f(x, \pi^{opt})) \leq g_{opt, Y_0}$. As a result, based on the first realization of the transmission delay $Y_0$, we can find an $x$ that makes the optimal time average age penalty smaller than the optimal query average age penalty. Then we can define this determination as a new update policy. This determination does not change the number of transmitted packet as it only modifies the first waiting time. Therefore, the same power constraint is satisfied. This completes the proof.

\section{Proof of Existence of the Limit for the PoW problem under any policy in $\Pi_{SD}^{UoW}$}\label{app:ExistenceOfLimitQueryAverageProblem}
\begin{proof}
In this proof, under a stationary and deterministic policy $\pi \in \Pi_{SD}^{UoW}$, we show that the limit $a_n$ exists as $n$ goes to $\infty$ where $a_n$ is the following:
\begin{align}\label{QueryAverageLimSup}
    a_n &\delequal \frac{E\big[\sum_{k=1}^n g(\Delta(Q_k)) \big| Y_0 \big]}{n}
\end{align}
\par $Y_0$ is given in the expectation, $Z_0 = z(Y_0) + x$ where x is constant; hence $Z_0$ is constant. Let $X_j = Y_j + Z_j$ for $j \in \mathbb{N}$. The transmission delays are i.i.d., and the update policy $\pi$ is a stationary and deterministic policy, which is a function of $Y_j$ for $j > 0$; thus $X_j$ is i.i.d.. The probabilities of $X_j$ can be calculated based on the probabilities of $Y_j$.
\begin{align}
    \Pr( X_j \in \mathcal{X}) = \Pr(Y_j \in \mathcal{Y}) \ \textrm{where} \ \mathcal{Y} = \bigcup_{x \in \mathcal{X}} \bigcup_{y \in \mathcal{Y} : y + z(y) = x} y
\end{align}
where $z(.)$ is the decision function of the stationary and deterministic policy. The request points can be represented as $R_{j+1} = Z_0 + \sum_{i=1}^j X_i$ for $j \geq 1$ and $R_1 = Z_0$. Let the stopping time $\tau_k = \min \{ \tau_k : R_{\tau_k} > Q_k\}$. The modulo operation is defined as the following:
\begin{equation}
    mod(R_{\tau_k}, T) = R_{\tau_k} - T \times \max\{k: Tk < R_{\tau_k}\}
\end{equation}
\par We can construct a Markov chain whose states are $\{mod(R_{\tau_k}, T), k \geq 1\}$ where $T$ is the period of query instants. $E[g(\Delta(Q_k)) | mod(R_{\tau_k}, T) ]$ can be calculated from the conditional expectation of $X_j$ given $Y_j$, independent of $k$. Throughout the proof, we consider $X_j$ in three different scenarios similar to \cite{robbins_1953}. 
\par \underline{The first scenario} is $X_j = 0$ for all $j \in \mathbb{N}$. It cannot be the case for this problem since $Y_j \geq B_L > 0$. 
\par \underline{The second scenario} is $\Pr(X_j \in \{ k \beta \colon B_L \leq k\beta \leq B_U + M \ \text{and} \ k \in \mathbb{N}\}) = 1$ such that $\beta$ is a rational multiple of $T$. In this scenario, the markov chain has a finite number of states. Let these states be $1,2, \dots, N$ These states communicate with each other. Therefore, it has a steady state distribution \cite[Section~4.3.1]{gallager_2013} and the limiting time-average fraction of time spent in each state can be calculated from \cite[Theorem~7.2.6]{gallager_2013}. Let these fractions be $p_1, p_2, \dots, p_N$. Then all the subsequences of $a_n$ goes to the same value equal to $\sum_{i=1}^N E[g(\Delta(Q_k)) | mod(R_{\tau_k}, T) = i^{th} \textrm{state}] \times p_i$. Since all the subsequences of $a_n$ goes to the same limit, the limit of $a_n$ exists. 
\par \underline{The third scenario} is all the random variables $X_j$ except the previous scenarios. Given $R_j = mT + a$ where $m \in \mathbb{N}$ and $a < T$, the probability of which $R_{j}$ is a stopping time is equal to $\Pr(X_{j-1} >a)$. From \cite[Theorem~1 and 2]{robbins_1953}, $R_j$ is equidistributed in modulo $T$ with probability 1. Due to the equidistriution, the limiting time-average fraction of time spent in the state of $mod(R_{\tau_k}, T) = a$ exists and it is proportional to $\Pr(X_j > a)$. Once the limiting time-average fraction exists, all the subsequences of $a_n$ goes to the same value similar to the second scenario. Thus, the limit of $a_n$ exists. This completes the proof.

\end{proof}
\section{Proof Of Uniform Convergence}\label{app:UniformCon}
\begin{proof}

Let $f_k \colon \mathbb{N} \rightarrow \mathbb{R}$, $k \in \mathbb{N}$ be a function  such that
\begin{equation}
    f_k(n) = \frac{\mathlarger{\sum_{i = 0}^{Q_n/2^{-k}}} g\bigg(\Delta_{\pi^{opt}, \mathbf{Y}} \big(i \times 2^{-k}\big)\bigg)}{Q_n/2^{-k}}
\end{equation}
Let $f \colon \mathbb{N} \rightarrow \mathbb{R}$ be a function such that
\begin{equation}
    f(n) = \frac{\int^{Q_n}_{0} g\big(\Delta_{\pi^{opt}, \mathbf{Y}}( t)\big) dt}{Q_n}
\end{equation}
\par If we prove that $f_k \rightarrow f$ uniformly for every $\{Y_j\}$ and $\{Z_j\}$ sequence providing that $Y_j \in [B_L,B_U]$ and $Z_j \in [0,M]$, we can ignore the expectation since it is uniformly convergent for every possible sequence. Then, the proof is completed.
\par Let $M_k \in \mathbb{R}$ be
\begin{equation}
    M_k = \sup_{n \in \mathbb{N}}|f_k(n) - f(n)|
\end{equation}
As penalty function $g$ is non-decreasing and $\Delta_{\pi^{opt}, \mathbf{Y}}(.) < B_U + M$, then $M_k \leq g(M+B_U)$. Furthermore, as $k$ goes to infinity, $M_k$ approaches $0$ due to the continuity of the penalty function, $g$. As a result, $f_k$ is uniformly convergent to $f$ by \cite[Theorem~7.9]{rudin_1976}. This completes the proof.

\end{proof}


\ifCLASSOPTIONcaptionsoff
\newpage
\fi



%
	
	
\bibliography{references} 
%







\end{document}